\documentclass[a4paper,11pt]{article}
\usepackage[T1]{fontenc}
\usepackage[latin9]{inputenc}
\usepackage{amsmath}
\usepackage{amssymb}
\usepackage{mathtools}
\usepackage{mathtools}
\usepackage{amsthm}

\usepackage{appendix}
\usepackage{hyperref}

\newtheorem{proposition}{Proposition}[section]
\newtheorem{lemma}{Lemma}[section]
\newtheorem{claim}{Claim}[section]
\newtheorem{theorem}{Theorem}[section]
\newtheorem{corollary}{Corollary}[section]

\makeatletter
\@ifundefined{date}{}{\date{}}
\makeatother
\usepackage{fullpage}
\usepackage{graphicx}
\graphicspath{{figures/}}

\usepackage{babel}
\usepackage{url}
\usepackage{algorithm} 
\usepackage{algpseudocode}
\usepackage{bbold}
\usepackage{natbib}

\begin{document}
\title{On the (In)approximability of Combinatorial Contracts\thanks{This material is based upon work partially supported by the National Science Foundation under Grant No. DMS-1928930 and by the Alfred P. Sloan Foundation under grant G-2021-16778, while Tomer was in residence at the Simons Laufer Mathematical Sciences Institute (formerly MSRI) in Berkeley, California, during the Fall 2023 semester. This project has received funding from the European Research Council (ERC) under the European Union's Horizon 2020 research and innovation program (grant agreement No. 866132), by an Amazon Research Award, and by the NSF-BSF (grant number 2020788).
}} 
\author{Tomer Ezra\thanks{Simons Laufer Mathematical Sciences Institute, USA. Email: tomer.ezra@gmail.com}\and Michal Feldman\thanks{Tel Aviv University, Israel. Email: mfeldman@tauex.tau.ac.il}\and Maya Schlesinger\thanks{{Tel Aviv University, Israel. Email: mayas1@mail.tau.ac.il}}}
\maketitle
\newcommand{\reals}{\mathbb{R}}
\newcommand{\smallClique}{\text{SMALL}}
\newcommand{\largeClique}{\text{LARGE}}
\thispagestyle{empty}

\begin{abstract}
    We study two recent combinatorial contract design models, which highlight different sources of complexity that may arise in contract design, where a principal delegates the execution of a costly project to others. 
In both settings, the principal cannot observe the choices of the agent(s), only the project's outcome (success or failure), and incentivizes the agent(s) using a contract, a payment scheme that specifies the payment to the agent(s) upon a project's success. 
We present results that resolve open problems and advance our understanding of the computational complexity of both settings.

In the {\em multi-agent} setting,  
the project is delegated to a team of agents, where each agent chooses whether or not to exert effort. 
A success probability function maps any subset of agents who exert effort to a probability of the project's success. 
For the family of submodular success probability functions, D{\"u}tting et al. [2023] established a poly-time constant factor approximation to the optimal contract, and left open whether this problem admits a PTAS. We answer this question on the negative, by showing that no poly-time algorithm guarantees a better than $0.7$-approximation to the optimal contract.
For XOS functions, they give a poly-time constant approximation with value and demand queries. We show that with value queries only, one cannot get any constant approximation. 

In the {\em multi-action} setting,  
the project is delegated to a single agent, who can take any subset of a given set of actions.
Here, a success probability function maps any subset of actions to a probability of the project's success. 
D{\"u}tting et al. [2021a] showed a poly-time algorithm for computing an optimal contract for gross substitutes success probability functions, and showed that the problem is NP-hard for submodular functions. We further strengthen this hardness result by showing that this problem does not admit any constant approximation. Furthermore, for the broader class of XOS functions, we establish the hardness of obtaining a $n^{-1/2+\varepsilon}$-approximation for any $\varepsilon > 0$.
\end{abstract}

\newpage
\pagenumbering{arabic} 

\section{Introduction}
Contract theory is a pillar in microeconomics, studying how to incentivize agents to exert costly effort when their actions are hidden.
This problem is explored using the %hidden action 
principal-agent model introduced by \citet{holmstrom1979moral} and \citet{grossman-hart-1983}.
In this model, a principal wishes to delegate the execution of a costly task to an agent who can take one of $n$ actions, each associated with a cost and a probability distribution over outcomes. The agent's action is hidden from the principal, who can observe only the realized outcome. 
To incentivize the agent to exert effort, the principal designs a contract, which is a payment scheme that specifies a payment for every possible outcome.
The goal of the principal is to find a contract that maximizes her utility (expected reward minus expected payment), assuming the agent takes the action that maximizes his own utility (expected payment minus cost). 
This problem can be solved in polynomial time using linear programming \citep{grossman-hart-1983}.

In recent years, the principal-agent model has been extended to combinatorial settings along different dimensions, such as multiple agents \citep{BabaioffFNW12,EmekF12,multi-agent-contracts}, multiple actions \citep{combinatorial-contracts} and exponentially many outcomes~\citep{dutting2021complexity}.

In this work, we study two of these combinatorial contract models, namely the multi-agent and multi-action settings. 
%one with multiple agents and one with multiple actions. 
%introduced by \citep{combinatorial-contracts}, and \citep{BabaioffFNW12,multi-agent-contracts}. 
In both of these models, the focus is on the case of binary outcome, where the project can either succeed or fail, and the principal receives some reward (normalized to $1$) if the project succeeds.
%the outcome is binary \tee{(as was the focus in previous works)}, where the project can either succeed or fail, and the principal receives some reward (normalized to $1$) if the project succeeds. 
Notably, finding an (approximate) optimal contract in the binary-outcome model, is equivalent to finding an (approximate) optimal linear contract in settings with more than two outcomes. 
% \mse{As observed in previous work, in the binary-outcome model linear contracts are optimal. It is known that finding the optimal contract in the binary model is The binary-outcome model is common in the literature and allows us to focus on linear contracts that are known to be optimal in this case. In addition, since all of our results are negative, this only strengthens them. }
% In this case, the optimal contract is linear (i.e., the principal pays the agents a constant fraction of the reward), and approximating the best contract is equivalent to approximating the optimal linear contract. 
Thus, we restrict attention to linear contracts without loss of generality. Moreover, since our focus is on hardness results, restricting attention to the binary-outcome case only strengthens our results.

\vspace{0.1in}
{\bf Setting 1: Multi-agent.} 
In the multi-agent model \citep{BabaioffFNW12,multi-agent-contracts}, the principal delegates the execution of a costly project to a team of $n$ agents. Every agent can either exert effort (at some cost to the agent) or not. 
At the heart of the model is a success probability function $f:2^{[n]}\rightarrow [0,1]$, which specifies, for every subset of agents who exert effort, the probability that the project succeeds. 

The principal incentivizes the agents through a contract that specifies for every agent $i$, a non-negative payment $\alpha_i$ that the principal pays the agent if the project succeeds. 
The principal's utility is defined as her expected reward minus the expected sum of payments to the agents.
Given a contract, an agent's utility is defined as his expected payment from the contract minus his cost if he chooses to exert effort. 
Thus, a contract by the principal induces a game between the agents, and we consider the agent actions in an equilibrium of the game. 
The principal's goal is to maximize her expected utility in equilibrium.

For {\em submodular} success probability functions, \citet{multi-agent-contracts} devise a constant factor approximation algorithm, using value query access (a value query receives a set $S\subseteq [n]$ and returns the value $f(S)$).
They left as an open problem whether the problem admits a PTAS. 
For the larger class of {\em XOS} success probability functions, they also give a constant approximation algorithm, using both value and demand queries (a demand query receives a price vector
$p \in \reals_{\geq 0}^n$, and returns a set $S$ that maximizes $f(S) -\sum_{i \in S} p_i$). 
For the XOS class, they show that it is not possible to obtain a better-than-constant approximation with value and demand queries.

\vspace{0.1in}
{\bf Setting 2: Multi-action.} 
In the multi-action model \citep{combinatorial-contracts}, the principal delegates the execution of the project to a single agent, who can take any subset of $n$ possible actions. Each action $i$ is associated with a cost $c_i$, and when the agent executes a set of actions $S \subseteq [n]$, he incurs the sum of their costs. 
Here, the success probability function $f:2^{[n]}\rightarrow [0,1]$ maps any subset of the actions to a success probability of the project. 

In this model, the principal specifies a single non-negative payment $\alpha$ that is paid to the agent if the project succeeds. 
The agent then chooses a subset of actions that maximizes his utility (the expected payment from the principal minus the cost he incurs). The principal's utility is the expected reward minus the expected payment to the agent.

\citet{combinatorial-contracts} show that computing an optimal contract for
submodular success probability functions
is NP-hard, and left as an open question whether there exists an approximation algorithm for the problem, for submodular success probability functions, as well as for the larger classes of XOS and subadditive functions. 
(For the class of gross-substitutes success probability functions --- a strict subclass of submodular functions --- they devise a polytime algorithm for computing an optimal contract, using access to a value oracle.)

\subsection{Our Results}
%In this paper, we tackle open problems raised in \citep{multi-agent-contracts} and \citep{combinatorial-contracts}.

\paragraph{Setting 1: Multi-agent.}
Our first set of results concern the multi-agent setting. 
The first result resolves the open question from \citep{multi-agent-contracts} in the negative, showing that the multi-agent problem with  submodular success probability functions does not admit a PTAS. 

\vspace{0.1in}
\noindent {\bf Theorem (multi-agent, submodular):} 
    In the multi-agent model, with  %presented in \citet{multi-agent-contracts}, for 
    submodular success probability function, no polynomial time algorithm with value oracle access can approximate the optimal contract to within a factor of $0.7$, unless P=NP.
\vspace{0.1in}

We then turn to XOS success probability functions. 
\citet{multi-agent-contracts} provide a poly-time constant approximation for XOS with demand queries. 
We show that no algorithm can do better than $O(n^{-1/6})$-approximation with poly-many value queries, thus establishing a separation between the power of value and demand queries for XOS functions. 

\vspace{0.1in}
\noindent {\bf Theorem (multi-agent, XOS):} 
    In the multi-agent model, with XOS success probability function, no (randomized) algorithm that makes poly-many value queries can approximate the optimal contract (with high probability) to within a factor greater than $4n^{-1/6}$.

\paragraph{Setting 2: Multi-action.}
Our second set of results consider the 
multi-action model. 
We first show that obtaining  any constant approximation for submodular functions is hard. 

\vspace{0.1in}
\noindent {\bf Theorem (multi-action, submodular):} 
    In the multi-action model, 
    %presented in \citet{combinatorial-contracts}, for 
    with submodular success probability function $f$, no polynomial time algorithm with value oracle access can approximate the optimal contract to within a constant factor, unless P=NP.
\vspace{0.1in}

We then show that for the broader class of XOS success probability functions, it is hard to obtain a $n^{-1/2+\varepsilon}$-approximation for any $\varepsilon > 0$.

\vspace{0.1in}
\noindent {\bf Theorem (multi-action, XOS):} 
    In the multi-action model, 
    %presented in \citet{combinatorial-contracts}, for 
    with XOS success probability function, under value oracle access, for any $\varepsilon > 0$, no polynomial time algorithm with value query access can approximate the optimal contract to within a factor of $n^{-\frac{1}{2}+\varepsilon}$, unless P=NP.
\vspace{0.1in}

\subsection{Our Techniques}

\noindent {\bf Submodular functions.} Both of our hardness results for submodular functions are based on an NP-hard promise problem for normalized unweighted coverage functions (a subclass of submodular), which is a generalization of a result by \citet{feige}. We introduce this problem and prove its hardness in Section~\ref{sec: coverage promise problem} and Appendix~\ref{app: coverage promise problem proof}.

% In particular, we show that it is NP-hard to distinguish between a normalized unweighted coverage function $f$ that has a relatively small set $S$ with $f(S)=1$, and one for which a significantly larger set would be required to get close to $1$.
% This leads to our submodular hardness results as follows: by setting uniform costs to all actions/agents, an approximately optimal contract can distinguish between the case where a relatively small set $S$ achieves $f(S)=1$ (which requires a relatively small payment) and the case where a larger set is required to get close to $1$ (which requires a relatively large payment). \msc{Do we want to change this paragraph? }\tec{yes}
In particular, we show that it is NP-hard to distinguish between a normalized unweighted coverage function $f$ that has a relatively small set $S$ with $f(S)=1$, and one for which a significantly larger set would be required to get close to $1$.
This leads to our submodular hardness results as follows: when setting uniform costs to all actions/agents, an optimal contract has a significantly different utility to the principal in each case. When a relatively small set $S$ achieves $f(S)=1$, the principal's utility is relatively high, and when a significantly larger set is required to get close to $1$, the principal's utility is relatively low. 
In the multi-agent setting, the principal's utility from a given contract is easy to compute, which allows an approximately optimal contract to distinguish between the two cases. 

% In the multi-action setting, where the principal's utility from a given contract isn't necessarily easy to compute, our reduction also involves the addition of a new action. We normalize the given function by $\frac{1}{2}$, and this new action has an additional additive reward of $\frac{1}{2}$. The cost of this new action is determined by the minimum utility of any approximately optimal contract in the first case (when a relatively small set $S$ achieves $f(S)=1$). Namely, if we know that in the first case any approximately optimal contract has a utility greater than $\beta$, this new action will a cost of $\frac{1}{2}(1-\beta)$. For any contract $t < 1-\beta$, the agent is incentivized not to take this action, but for any contract $t \ge 1-\beta$, his best response will always include it. Thus - in the first case any approximately optimal contract is less than $1-\beta$ (since the utility from a contract $t\ge1-\beta$ is at most $1-t \le \beta$), and in the second case (assuming a large enough utility gap), any approximately optimal contract is at least $1-\beta$. 

In the multi-action setting, where the principal's utility from a given contract isn't necessarily easy to compute, our reduction also involves the addition of a new action. This new action is defined in such a way that only a high payment to the agent incentivizes the agent to take it. Thus - the principal has to decide whether to choose the optimal contract for the original problem, or to incentivize the agent to take this new action, which leads to a separation of contracts between the two cases.

% We normalize the given function by $\frac{1}{2}$, and this new action has an additional additive reward of $\frac{1}{2}$. The cost of this new action is determined by the minimum utility of any approximately optimal contract in the first case (when a relatively small set $S$ achieves $f(S)=1$). Namely, if we know that in the first case any approximately optimal contract has a utility greater than $\beta$, this new action will a cost of $\frac{1}{2}(1-\beta)$. For any contract $t < 1-\beta$, the agent is incentivized not to take this action, but for any contract $t \ge 1-\beta$, his best response will always include it. Thus - in the first case any approximately optimal contract is less than $1-\beta$ (since the utility from a contract $t\ge1-\beta$ is at most $1-t \le \beta$), and in the second case (assuming a large enough utility gap), any approximately optimal contract is at least $1-\beta$. 

\vspace{0.1in}
\noindent {\bf Multi-action, XOS functions.}
For this result, we construct a reduction from the problem of approximating the size of the largest clique in a graph to our problem. 
We rely on the hardness result of \citet{Hastad_1999, zuckerman2006} for approximating the largest clique in a graph $G$, denoted by $\omega(G)$. In particular, \citet{Hastad_1999, zuckerman2006} show that there is no poly-time algorithm that approximates $\omega(G)$ within a factor of $n^{-1+\varepsilon}$ (for any $\varepsilon > 0$) unless P=NP.
We show that given a $\beta$-approximation algorithm (for $\beta \in (0,1)$) for the optimal contract, one can approximate $\omega(G)$ within a factor of $\beta^2/4$, which implies our inapproximability of $n^{-1/2+\varepsilon}$ for the optimal contract. To achieve this, for any parameter $\beta \in (0,1)$, we give an algorithm that on input $(G, \delta)$, where $G$ is a graph and $\delta \in \mathbb{N}^+$, creates an instance of the multi-action contract problem with an XOS success probability function, for which value queries can be computed in polynomial time. In the constructed instance there are only two ``reasonable'' candidates for a contract, regardless of the structure of $G$; these are the values of $\alpha$ at which the agent's best response may change. The lower of these candidates is better than the other by at least a factor of $\beta$ when $\omega(G) \le \delta$, and the reverse is true when $\omega(G) \ge 2\delta/\beta^2$. This gives us the ability to distinguish between the case where $\omega(G) \le \delta$ and $\omega(G) \ge 2\delta/\beta^2$. 
By repeating this process for different values of $\delta$ we can approximate $\omega(G)$ within a factor of $\beta^2/4$.

\vspace{0.1in}
\noindent {\bf Multi-agent, XOS functions.}
Our inapproximability result for this case is information theoretic, and relies on ``hiding'' a good contract, so that no algorithm with poly-many value queries can find it with non-negligible probability.
In particular, for any $n$, we choose a set $G\subseteq A$ of $m = n^{1/3}$ ``good'' agents uniformly at random. 
We define an XOS success probability function such that sets of size $O(m)$ may have a high success probability only if they have a large intersection with $G$, and any value query reveals negligible information regarding the set $G$. 
We set equal costs such that incentivizing more than $2m$ agents becomes unprofitable to the principal. 
Thus, in order to get a good approximation, the algorithm must find a relatively small set of agents that has a large intersection with $G$.
Since our construction of the success probability function is such that value queries reveal negligible information on $G$, the algorithm has a negligible probability of finding such a set.

\subsection{Related Work}

\paragraph{Multi-agent settings: additional related work.}
\citet{BabaioffFNW12} introduced a multi-agent model where every agent decides whether to exert effort or not, and succeeds in his own task (independently) with a higher probability if he exerts effort. 
The project's success is then a function of the individual outcomes by the agents. 
They show that computing the optimal contract in this model is \#P-hard in general, and provide a polytime algorithm for the special case where the project succeeds iff all agents succeed in their individual tasks (AND function). 
\citet{EmekF12} show that computing the optimal contract in the special case where the project succeeds iff at least one agent succeeds (OR function) is NP-hard, and provide an FPTAS for this problem.

\citet{multi-agent-contracts} extend the model of \citep{BabaioffFNW12} to the model presented in our paper, where the project's outcome is stochastically determined by the set of agents who have exerted effort, according to a success probability function $f: 2^{[n]}\rightarrow [0,1]$. 
Their primary result is the development of a constant approximation algorithm for XOS success probability functions. They complement this result by showing an upper bound of $O(1/\sqrt{n})$ for subadditive functions, and an upper bound of a constant for XOS functions. They also show that the problem is NP-hard even for additive functions, and devise an FPTAS for this case.

\citet{vuong2023supermodular} study the model of \citep{multi-agent-contracts} for the case where the function $f$ is supermodular, and show that no polynomial time algorithm can achieve any constant approximation nor an additive FPTAS.
They also present an additive PTAS for a special case of graph-based super modular valuations.

\citet{castiglioni2023multiagent} study a multi-agent setting in which each agent has his own outcome, which is observable by the principal, and the principal's reward depends on all the individual outcomes. 
When the principal's reward is supermodular, they show that it is NP-hard to get any constant approximation to the optimal contract. They also give a poly-time algorithm for the optimal contract in special cases. 
When the principal's reward is submodular, they show that for any $\alpha \in (0,1)$ it is NP-hard to get a $n^{\alpha - 1}$-approximation, and they also provide a poly-time algorithm that gives a $(1-\frac{1}{e})$-approximation up to a small additive loss. 

\citet{dasaratha2023equity} consider a multi-agent setting with graph-based reward functions, and continuous effort levels, and characterize the optimal equilibrium induced by a linear contract.

\paragraph{Multi-action settings: additional related work. } 
\citet{vuong2023supermodular} and \citet{duetting2023combinatorial} further explore the multi-action model of \citep{combinatorial-contracts}. 
They present a poly-time algorithm for computing the optimal contract for any class of instances that admits an efficient algorithm for the agent's demand and poly-many ``breakpoints" in the agent's demand. A direct corollary of this result is a polynomial time algorithm for computing the optimal contract when the success probability function is supermodular and the cost function is submodular.
\citet{duetting2023combinatorial} further show a class of XOS success probability functions (matching-based) which admits an efficient algorithm for the agent's demand, but has a super-polynomial number of breakpoints in the agent's demand. Computing the optimal contract for this class remains an open problem. (Pseudo) polynomial algorithms are presented for two special cases.

\paragraph{Additional combinatorial contract models.}
Beyond multi-agent and multi-action, one can consider other dimensions in which a contracting problem grows. 
For example, \citet{dutting2021complexity} consider a setting with exponentially many outcomes. 
They show that under a constant number of actions, it is NP-hard to compute an optimal contract. They proceed to weaken their restriction on contracts,  
and consider ``approximate-IC'' contracts, in which the principal suggests an action for the agent to take (in addition to the payment scheme), and the agent takes it as long as its utility is not much lower than that of another possible action. They present an FPTAS that computes an approximate-IC contract that gives the principal an expected utility of at least that achieved in the optimal (IC) contract. For an arbitrary number of actions, they show NP-hardness of any constant approximation, even for approximate-IC contracts.

\paragraph{Contracts for agents with types.}
\citet{guruganesh2021contracts} consider a setting where in addition to hidden action, the agent also has a private type, which changes the effect of each action he takes on the project's outcome. In the private type setting the principal may wish to incentivize the agent with a menu of contracts, i.e., a set of contracts from which the agent is free to choose whichever contract he prefers.
They show APX-hardness of both the optimal contract and the optimal menu of contracts. In contrast, \citet{alon2021contracts} consider the case where the agent has a single-dimensional private type, and they present a characterization of implementable allocation rules (mappings of agent types to actions), which allows them to design a poly-time algorithm for the optimal contract with a constant number of actions.

\citet{castiglioni2022designing}  study the case where the agent's private type is Bayesian, i.e., drawn from some known finitely-supported distribution. They study menus of randomized contracts (defined as distributions over payment vectors), wherein upon the agent's choice of randomized contract, the principal draws a single deterministic contract from the distribution, and the agent plays his best response to this deterministic contract. They show that an almost optimal menu of randomized contracts can be computed in polynomial time. They also show that the problem of computing an optimal menu of deterministic contracts cannot be approximated within any constant factor in polynomial time, and that it does not admit an additive FPTAS.

\paragraph{Optimizing the efforts of others. }
Contract design is part of an emerging frontier in algorithmic game theory regarding optimizing the effort of others (see, e.g., the STOC 2022 TheoryFest workshop with the same title). 
In addition to contract design, this field includes recent work on strategic classification \citep{kleinberg2020classifiers, bechavod2022information}, delegation \citep{kleinberg2018delegated, bechtel2022delegated}, and scoring rules \citep{papireddygari2022contracts, hartline2022optimization}.

\section{Model and Preliminaries}\label{sec:model}
We first describe the basic version of the contract design problem, also known as the hidden-action or principal-agent setting.
We then present two extensions, one with multiple agents, the other with multiple actions.
For simplicity, we restrict attention to a binary-outcome setting, where a project either succeeds or fails.

\subsection{Basic Principal-Agent Setting} A single principal interacts with a single agent, in an effort to make a project succeed. The agent has a set $A$ of possible actions, each with associated cost $c_i \geq 0$ and probability $p_i \in [0,1]$. When the agent selects action $i\in A$, he incurs a cost of $c_i$, and the project succeeds with  probability $p_i$, and fails with probability $1-p_i$. If the project succeeds, the principal gets a reward which we normalize to $1$. The principal is not aware of which action the agent has taken, only if the project has succeeded or failed.

\paragraph{Contracts.}
Since exerting effort is costly and reaps benefits only to the principal, in and by itself the agent has no incentive to exert effort. This challenge is often referred to as ``moral hazard''.
To incentivize the agent to exert effort, the principal specifies a contract that maps project outcomes (in this case, ``success'' and ``failure'') to payments made to the agent by the principal. It is well known that in the binary-outcome case, it is without loss of generality to assume that the payment for failure is 0. Thus, a contract can be fully described by a parameter $\alpha \in [0,1]$, which is the fraction of the principal's reward that is paid to the agent (in our case, where the reward is normalized to $1$, $\alpha$ is essentially the payment for success).

Under a contract $\alpha \in [0, 1]$, the agent's utility from action $i\in A$ is the expected payment minus the cost, i.e.,
\[
u_A(\alpha, i) = p_i \cdot \alpha - c_i.
\]
The principal's utility under a contract $\alpha$ and an agent's action $i$ is the expected reward minus the expected payment, i.e.,
\[
u_P(\alpha, i) = p_i (1-\alpha).
\]
Given a contract $\alpha$, the agent's best response is an action that maximizes his utility, namely $i_{\alpha} \in \arg\max_{i \in A} u_A(\alpha,i)$.
As standard in the literature, the agent breaks ties in favor of the principal's utility.
The principal's problem, our problem in this paper, is to find the contract $\alpha$ that maximizes her utility, given that the agent best responds. 
Let $u_P(\alpha)=u_P(\alpha,i_{\alpha})$, then the principal's objective is to find $\alpha$ that maximizes $u_P(\alpha)$.

\subsection{Combinatorial Contract Settings}
\label{sec:comb-contracts-model}

In what follows, we define two combinatorial settings, one with multiple agent (introduced by \cite{BabaioffFNW12}, as presented in \cite{multi-agent-contracts}), the other with multiple actions (introduced by \cite{combinatorial-contracts}). In both cases we use a set function $f:2^{A}\rightarrow [0,1]$ that maps every subset of a set $A$ to a success probability. 
The set $A$ denotes the set of agents in the first setting, and the set of actions in the second setting. When considering $f$ outside of a specific context, we refer to $A$ as the set of items. 

We focus on success probability functions $f$ that belong to one of the following classes of complement-free set functions \citep{nisan2006communication}. A set function $f:2^A\rightarrow \reals_{\ge 0}$ is:
\begin{enumerate}
    \item \textit{Additive} if there exist values $v_1,\dots,v_n \in \reals_{\ge 0}$ such that $\forall S\subseteq A.~f(S) = \sum_{i\in S} v_i$.
    \item \textit{Coverage} if there is a set of elements $U$, with associated positive weights  $\{w_u\}_{u\in U}$, and a mapping $h:A \rightarrow 2^U$ such that for every $S\subseteq A$, $f(S)=\sum_{u\in U } w_u \cdot \mathbb{1}[\exists i \in S.~u \in  h(i)]$, where $\mathbb{1}[B]$ is the indicator variable of the event $B$.
    In this paper, we focus on a special case of coverage functions, called \textit{normalized unweighted coverage functions}, in which $w_u=\frac{1}{|U|}$ for every $u \in |U|$.
    We represent a normalized unweighted coverage function $f$ using a tuple $(U,A,h)$, for which $f(S) = \frac{1}{|U|}\left|\bigcup_{i\in S} h(i) \right|$.
    \item \textit{Submodular} if for any two sets $S, S'\subseteq A$ s.t. $S\subseteq S'$ and $i\in A$ it holds that $f(i\mid S) \ge f(i \mid S')$, where $f(i \mid S) = f(S \cup \{i\}) - f(S)$ is the marginal contribution of $i$ to $S$.
    \item \textit{XOS} if there exists a finite collection of additive functions $\{a_i: 2^A \rightarrow \reals_{\ge 0} \}_{i=1}^k$ such that for every $S\subseteq A$, $f(S) = \max_{i=1,\dots, k} a_i(S)$.
    %\item A set function $f:2^A\rightarrow \reals_{\ge 0}$ is \textit{subadditive} if for any two sets $S, S'\subseteq A$ it holds 
    %\[
    %f(S) + f(S') \le f(S\cup S').
    %\]
\end{enumerate}

It is well known that additive $\subset$ coverage $\subset$ submodular $\subset$ XOS, %$\subset$ subadditive 
and all containment relations are strict \citep{lehmann2001combinatorial}.

\paragraph{Computational model.}
Since the success probability function $f:2^A\rightarrow [0,1]$ contains exponentially many (in $|A|$) values, we assume, as is common in the literature, that the algorithm has a value oracle access, which, for every set $S\subseteq A$, returns $f(S)$. 
It should be noted that most of our results hold under an even stronger assumption. Namely, that the success probability function $f$ admits a succinct representation, for which a value oracle can be computed efficiently, and that this representation is given to the algorithm. This assumption implies that these results are purely computational hardness ones, as the algorithm essentially knows the entire function.

We next present the two combinatorial models considered in this paper.

\paragraph{Setting 1: Multiple agents. }
In the multiple agents setting, the principal interacts with a set $A$ of $n$ agents. Every agent $i \in A$  decides whether to exert effort or not (binary action). Exerting effort comes with a cost of $c_i \ge 0$ (otherwise the cost is zero).
The success probability function $f:2^A\rightarrow [0,1]$ maps 
every set of agents who exert effort to a success probability of the project, where $f(S)$ denotes the success probability if $S$ is the set of agents who exert effort.  

A contract is now a vector $\alpha =(\alpha_1,\ldots,\alpha_n)\in [0,1]^n$, where $\alpha_i$ is the payment to agent $i$ upon a project success.

Given a contract 
$\alpha=(\alpha_1,\ldots,\alpha_n)$ and a set $S$ of agents who exert effort, the principal's utility is given by $\left(1-\sum_{i\in A} \alpha_i\right)f(S)$. Agent $i$'s utility is given by $\alpha_i f(S) - \mathbb{1}[i\in S] c_i$. Note that agent $i$ is paid in expectation $\alpha_i f(S)$ regardless of whether he exerts effort or not, but pays $c_i$ only if he exerts effort (i.e., if $i\in S$).

To analyze contracts, we consider the (pure) Nash equilibria of the induced game among the agents. A contract $\alpha=(\alpha_1,\ldots,\alpha_n)$ is said to incentivize a set $S\subseteq A$ of agents to exert effort (in equilibrium) if 
\begin{eqnarray*}
\alpha_i f(S) -  c_i  \ge \alpha_i f(S \setminus \{i\}) & &\text{for all }i\in S\text{, and}\\
\alpha_i f(S) \ge \alpha_i f(S\cup \{i\}) -  c_i & & \text{for all }i\notin S.
\end{eqnarray*}

Since equilibria may not be unique, we think of a contract as a pair $(\alpha, S)$ where $S$ is a set of agents incentivized to exert effort (in equilibrium).

It is easy to observe that 
for any set $S\subseteq A$, the best way for the principal to incentivize the agents in $S$ is by the contract 
\[
\begin{array}{lc}
\alpha_i = \frac{c_i}{f(i \mid S \setminus \{i\})} & \text{for all }i\in S\text{, and}\\
\alpha_i = 0 & \text{for all }i\notin S,
\end{array}
\]
where $f(i \mid S \setminus \{i\})=f(S) - f(S\setminus \{i\})$ is the marginal contribution to $S$ of adding $i$ to $S \setminus \{i\}$. We interpret $\frac{c_i}{f(i \mid S \setminus \{i\})}$ as $0$ if $c_i = 0$ and $f(i \mid S \setminus \{i\})=0$ and as infinity when $c_i > 0$ and $f(i \mid S \setminus \{i\})=0$. The principal thus tries to find a set $S$ that maximizes $g(S)$ where 
\[
g(S) = \left(1-\sum_{i\in S} \frac{c_i}{f(i \mid S\setminus \{i\})}\right)f(S).
\]

Let $S^\star$ be the optimal set of agents, i.e., the set that maximizes $g$. We say that $S$ is a $\beta$-approximation to the optimal contract (where $\beta \leq  1$) if $ g(S) \ge \beta \cdot g(S^\star)$.

\paragraph{Setting 2: Multiple actions. }
In the multiple actions setting, the principal interacts with a single agent, who faces a set $A$ of $n$ actions, and can choose any subset $S \subseteq A$ of them. 
Every action $i \in A$ is associated with a cost $c_i \ge 0$, and the cost of a set $S$ of actions is $\sum_{i\in S} c_i$. The success probability function $f(S)$ denotes the probability of a project success when the agent chooses the set of actions $S$. 

A contract is defined by a single parameter $\alpha \in (0,1)$, which denotes the payment to the agent upon the project's success. 
Given a contract $\alpha$, the agent's  and principal's utilities under a set of actions $S$ are, respectively,
\[
u_A(\alpha, S) = f(S) \cdot \alpha - \sum_{i\in S} c_i \quad\quad \mbox{and} \quad\quad u_P(\alpha, S) = f(S) (1-\alpha).
\]

The agent's best response for a contract $\alpha$ is $S_{\alpha} \in \arg\max_{S \subseteq A} u_A(\alpha, S)$. 
As before, the agent breaks ties in favor of the principal's utility. We also denote $u_P(\alpha)=u_P(\alpha,S_{\alpha})$, and the principal's objective is to find a contract $\alpha$ that maximizes her utility $u_P(\alpha)$.
We denote by $\alpha^\star$ the optimal contract, i.e., the contract that maximizes $u_P(\alpha^\star)$. We say that a contract $\alpha$ is a $\beta$-approximation (where $\beta \leq 1$) if $u_P(\alpha) \ge \beta \cdot u_P(\alpha^\star)$.

\section{An NP-hard Promise Problem of Coverage Functions} \label{sec: coverage promise problem}
In this section we define a promise problem regarding normalized unweighted coverage functions that is the basis of our hardness results for the contract models with submodular success probability functions. Essentially, we show that it is NP-hard to distinguish between a normalized unweighted coverage function $f$ that has a relatively small set $S$ with $f(S)=1$, and one for which a significantly larger set $T$ would be required to get close to $f(T)=1$. This naturally leads to our hardness results in Sections~\ref{sec:multi-agent model hardness} and~\ref{sec: multi-action-submodular}; by setting uniform costs to all actions / agents, an approximately optimal contract can usually distinguish between the case where a relatively small set achieves $f(S)=1$ (and incentivizing costs less to the principal) and a larger set is required to get close to $f(T)=1$ (and incentivizing costs more to the principal).

This hardness result is an extension of a hardness result by \citet{feige}, which we present next for completeness.
Recall that a normalized unweighted coverage function $f$ is given by a tuple $(U,A,h)$, where $U$ is a set of elements, and $h$ is a mapping from $A$ to $2^U$ (see Section~\ref{sec:comb-contracts-model}).

\begin{proposition}[\citep{feige}] \label{prop:feige}
    For every $0 < \varepsilon < e^{-1}$, on input $(k,f)$, where $k\in \mathbb{N}$ and $f=(U, A, h)$ is a normalized unweighted coverage function such that exactly one of the following two conditions holds:\vspace{-0.08in}
    \begin{enumerate}
    \itemsep-0.1em
        \item There exists a set $S\subseteq A$ of size $k$ such that $f(S) = 1$. 
        \item Every set $S\subseteq A$ of size $k$ satisfies $f(S) \le 1-e^{-1}+\varepsilon$.
    \end{enumerate}
    It is NP-hard to determine which of the two conditions is satisfied by the input.
\end{proposition}

Remark: \citet{feige} used a different terminology, but proved an equivalent result. 

We next present the following extension to Proposition~\ref{prop:feige}.
\begin{proposition} \label{prop:coverage hardness}
    For every $M > 1$ and every $0 < \varepsilon < e^{-1}$, on input $(k,f)$, where $k\in \mathbb{N}$ and $f=(U, A, h)$ is a normalized unweighted coverage function such that $\forall i\in A.~f(\{i\}) = \frac{1}{k}$ and exactly one of the following two conditions holds:
    \begin{enumerate}
        \item There exists a set $S\subseteq A$ of size $k$ such that $f(S) = 1$.
        \item Every set $S\subseteq A$ of size $\beta k$ such that $\beta \le M$ satisfies $f(S) \le 1-e^{-\beta}+\varepsilon$.
    \end{enumerate}
    It is NP-hard to determine which of the two conditions is satisfied by the input.
\end{proposition}

Proposition~\ref{prop:feige} asserts that it is hard to approximate the maximum value of $f(S)$ for sets of size $k$ within some factor, and in Proposition~\ref{prop:coverage hardness} we generalize this to the maximum value of $f(S)$ over all sets $S$ of some fixed size $\ell \in O(k)$. This is a necessary adjustment for our contract design problem, as we are not restricted to sets of size exactly $k$. 
Indeed, we can have either smaller or slightly larger sets (at smaller or slightly larger costs, respectively).

The proof of Proposition~\ref{prop:coverage hardness} is deferred to Appendix~\ref{app: coverage promise problem proof}.  
Like the proof of Proposition~\ref{prop:feige}, the proof is based on a reduction from the 
3CNF-5 separation problem defined in Proposition~\ref{prop: 3sat-5 hardness}
(known to be NP-hard), which we describe in Appendix~\ref{subsec:max-3sat-5}. The reduction relies on a $k$-prover proof system for 3CNF-5 formulas introduced by \citet{feige}, which we describe in Appendix~\ref{subsec:k-prover system}. In Appendix~\ref{subsec:construction of coverage function} we present the reduction, and in Appendix~\ref{subsec: coverage reduction proof} we prove its correctness, thus proving Proposition~\ref{prop:coverage hardness}.

\section{Hardness of Approximation for Multi-Agent, Submodular $f$}
\label{sec:multi-agent model hardness}

In this section, we settle an open problem
from \citep{multi-agent-contracts}. 
In particular, \citep{multi-agent-contracts} show that in a multi-agent setting, one can get constant-factor approximation for settings with submodular success probability function $f$, with value queries. 
It is left open whether one can get better than constant approximation for this setting. The following theorem resolves this question in the negative. %(Note that it is still open whether one can get better than constant for submodular functions, with demand queries.)
This is cast in the following theorem.

\begin{theorem}\label{thm:multi-agent hardness}
    In the multi-agent model, %presented in \citep{multi-agent-contracts}, 
    for submodular (and even normalized unweighted coverage) success probability function $f$, no polynomial time algorithm with value oracle access can approximate the optimal contract within a factor of $0.7$, unless $P=NP$. 
\end{theorem}

Before presenting the proof of this theorem, we recall some of the details of the multi-agent model (see Section~\ref{sec:model}). 
In this setting, the principal interacts with a set $A$ of $n$ agents. Every agent $i \in A$ decides whether to exert effort (at cost $c_i \ge 0$) or not. 
The success probability function $f:2^A\rightarrow [0,1]$ maps 
every set of agents who exert effort to a success probability of the project.
A contract is a vector $\alpha = (\alpha_1, \dots, \alpha_n)$ specifying the payment to each agent upon success. 
%Since the optimal contract $\alpha$ which incentivizes a set $S$ of agents to exert effort is 
% \[
% \begin{array}{lc}
% \alpha_i = \frac{c_i}{f(i \mid S \setminus \{i\})} & \text{for all }i\in S\text{, and}\\
% \alpha_i = 0 & \text{for all }i\notin S,
% \end{array}
% \]
The principal seeks to find the optimal set $S$ of agents to exert effort, which is equivalent to maximizing the function
\[
g(S) = \left(1-\sum_{i\in S} \frac{c_i}{f(i \mid S\setminus \{i\})}\right)f(S).
\]

\begin{proof}[Proof of Theorem~\ref{thm:multi-agent hardness}]
Our proof relies on Proposition~\ref{prop:coverage hardness}, by creating a reduction with the following properties:
Given as input $(k, f)$ where $f: 2^A\rightarrow [0,1]$ is a coverage function that satisfies one of the conditions of Proposition~\ref{prop:coverage hardness}, we construct an instance of the multi-agent contract problem with the following separation: 
Under condition (1) in the proposition, the principal's utility is at least $0.5$, whereas under condition (2), the principal's utility is strictly less than $0.35$.

Suppose we have a $0.7$-approximation algorithm for our contract problem, and let $(\alpha, S)$ be the output (i.e., contract) of this algorithm on our reduction-generated instance. If the principal's utility under $(\alpha, S)$ is greater than or equal to $0.35$ (note that computing the principal's utility is easy in this model), then we must be under condition (1) of Proposition~\ref{prop:coverage hardness}.
Otherwise (if the principal's utility is less than $0.35$), then we must be under condition (2), since for condition (1) to hold, we should get at least $0.7\cdot 0.5 = 0.35$.
Since the construction of our reduction is polynomial, this proves the theorem.

It remains to construct an instance that admits the separation above.

Given $k, f: 2^A\rightarrow [0,1]$ from Proposition~\ref{prop:coverage hardness} with $M=2, \varepsilon = 0.01$, we construct an instance of the multi-agent contract problem, where $A$ is the set of agents, $f$ is the success probability function, and $c_i = \frac{1}{2k^2}$ for every agent $i \in A$. We next establish the desired separation.
    
    \paragraph{Case 1: $f$ satisfies condition (1) from Proposition~\ref{prop:coverage hardness}.} Take set $S$ per condition (1) of Proposition~\ref{prop:coverage hardness} (i.e., $|S|=k$ and $f(S)=1$). We claim that for any $i\in S$, $f(i \mid S \setminus \{i\}) = \frac{1}{k}$. Indeed, 
    \[
    f(i \mid S \setminus \{i\}) = f(S) - f(S \setminus \{i\}) \ge 1  - (k-1)\frac{1}{k} = \frac{1}{k},
    \]
    where the inequality follows by $f(S)=1$ and submodularity of $f$. In the other direction we have $f(i\mid S \setminus \{i\}) \le f(\{i\}) = \frac{1}{k}$, giving $f(i\mid S \setminus \{i\}) = \frac{1}{k}$.
    We now have 
    \[
    \begin{split}
    g(S) &= \left(1-\sum_{i\in S}\frac{c_i}{f(i \mid S\setminus \{i\})}\right)f(S) = \left(1-\sum_{i\in S}k\cdot c_i\right)f(S) =\left(1-\frac{|S|}{2k}\right) \cdot 1=\frac{1}{2}.
    \end{split}
    \]
    We get that the principal's utility under the optimal contract is at least $0.5$, as desired.
    \footnote{We note that in this case, incentivizing $S$ is the optimal contract, since to incentivize each  agent the principal needs to pay at least $\frac{1}{2k}$, and since the success probability function  satisfies, for any $T\subseteq A$, $f(T) \le \frac{|T|}{k}$, we get that $g(T) \le (1-\frac{|T|}{2k}) \frac{|T|}{k}\le \frac{1}{2}$.}
    
    \paragraph{Case 2: $f$ satisfies condition (2) from Proposition~\ref{prop:coverage hardness}.} 
    Let $S\subseteq A$ be an arbitrary set. We show that $g(S) < 0.35$. 
    
    If $|S| \ge 2k$, we have 
    \[
    \begin{split}
    g(S) &= \left(1-\sum_{i\in S}\frac{c_i}{f(i \mid S\setminus \{i\})}\right)f(S) \le \left(1-\sum_{i\in S}\frac{c_i}{f(\{i\})}\right)f(S) = \left(1-\frac{|S|}{2k}\right) f(S) \le 0,
    \end{split}
    \] 
    where the first inequality is by submodularity of $f$, and the last inequality is by $|S| \ge 2k$.
    
    If $|S| < 2k = Mk$, we can apply condition (2) and get
    \[
    \begin{split}
    g(S) &\le \left(1-\sum_{i\in S}\frac{c_i}{f(\{i\})}\right)f(S) \le \left(1-\frac{|S|}{2k}\right)\left(1-e^{-\frac{|S|}{k}}+0.01\right) < 0.35,
    \end{split}
    \]
    where the third inequality is since $(1-\frac{x}{2})(1-e^{-x}+0.01) < 0.35$ for all $x$. This concludes the proof.
\end{proof}

In Appendix~\ref{app:approach}, we discuss the differences in our approach for hardness results, and the approach of \citet{multi-agent-contracts}, and explain why the hardness results in \citep{multi-agent-contracts} cannot be extended to show hardness of getting a PTAS for submodular success probability functions.

\section{Hardness of Approximation for Multi-Agent, XOS $f$}
In this section, we show that, in the multi-agent model, one cannot approximate the optimal contract under XOS success probability functions within a constant factor, with access to a value oracle. More formally, we prove the following theorem:

\begin{theorem}
    In the multi-agent model, with XOS success probability functions, no (randomized) algorithm that makes poly-many value queries can approximate the optimal contract (with high probability) to within a factor greater than $4n^{-1/6}$.
\end{theorem}
To prove this theorem, we define a probability distribution over XOS success probability functions and show an upper bound on the expected performance of any deterministic algorithm. By Yao's principle, this gives us an upper bound on the worst-case performance of a randomized algorithm, thus proving the theorem.

We note that unlike our other results, which are computational hardness results, this result is information theoretical, as we rely on the algorithm ``not knowing'' something about the success probability function.

\begin{proof}
For any $n$, we uniformly at random sample a subset $G\subseteq A=[n]$ of the agents of size $|G| = m=n^{1/3}$, and build the following instance:
\begin{enumerate}
    \item Our success probability function is $f_G: 2^A \rightarrow [0,1]$, and for any non-empty set $S\subseteq A$:
    \[
    f_G(S) = \frac{1}{n}\max\left( | S \cap G |, \sqrt{m}, \frac{|S|}{\sqrt{m}}\right).
    \]
    \item For any $i\in A$ the cost associated with agent $i$ is $c_i = \frac{1}{2m\cdot n}$.
\end{enumerate}
Note that $f_G$ is clearly XOS, as it is the maximum of three XOS functions (two of them are additive and one of them is unit demand), and the family of XOS is closed under maximum (in contrast to submodular).

We start by noting that
\[
g(G) = \left(1-\sum_{i\in G}\frac{c_i}{f_G(i \mid G \setminus \{i\})}\right)f_G(G) = \left(1-m\cdot \frac{n}{2m\cdot n}\right) \cdot \frac{m}{n} = \frac{m}{2n}.
\]

We call a value query on set $S$ \textit{successful} if $|S| \le m^{1.5}$ and $|S \cap G| > \sqrt{m}$, and \textit{unsuccessful} otherwise. 

We start by noting that the probability of any specific value query being successful is negligible:
\begin{lemma} \label{lemma: successful negligible}
    Let $S\subseteq A$. For $n\ge 512$, it holds that 
    \[
    Pr_G[S \text{ is successful}] \le e^{-\frac{\sqrt{m}}{4}}.
    \]
\end{lemma} 
\begin{proof}
    The probability of $S$ being successful when $G$ is chosen uniformly at random from all subsets of $A$ of size $m$ is $$ \Pr_G[|S \cap G| > \sqrt{m} \wedge |S| \leq m^{1.5}]$$ which is monotone in the size of $S$ up to size $m^{1.5}$. For $|S|=m^{1.5}$, $|S \cap G|$ is distributed as a hyper-geometric random variable $HG(n,m^{1.5},m)$.  Therefore 
    \[
    \begin{split}
    Pr_G[S \text{ is successful}] &\le Pr_{X\sim HG(n,m^{1.5},m)}[ X > \sqrt{m}] \le Pr_{X\sim Bin(m,\frac{m^{1.5}}{n-m})}[ X > \sqrt{m}] \\
    &\le Pr_{X\sim Bin(m,\frac{2m^{1.5}}{n})}[ X > \sqrt{m}],
    \end{split}
    \]
    where the second inequality is because the probability of success in each one of the $m$ draws (without replacement) in the hyper-geometric distribution $HG(n, m^{1.5}, m)$ is always at most $\frac{m^{1.5}}{n-m}$, and the last inequality is because for $n\ge 512$ it holds that $\frac{m^{1.5}}{n-m}\le \frac{2m^{1.5}}{n}$.
    
    When denoting $\mu = E_{X\sim Bin(m, \frac{2m^{1.5}}{n})}[X]= \frac{2m^{2.5}}{n}=\frac{2}{\sqrt{m}}$ and $\delta = \frac{\sqrt{m}}{2\mu} = \frac{m}{4}$ we can apply the Chernoff bound and get
    \[
    \begin{split}
    \Pr_{X\sim Bin(m, \frac{2m^{1.5}}{n})}[X > \sqrt{m}] &\le \Pr_{X\sim Bin(m, \frac{2m^{1.5}}{n})}[X \ge (1+\delta) \mu ] \\
    &\le e^\frac{-\delta^2 \mu}{2+\delta} \le e^\frac{-\delta^2 \mu}{2\delta} = e^{-\frac{m^{1.5}/8}{m/2}} = e^{-\frac{\sqrt{m}}{4}}. \qedhere
    \end{split} 
    \]
\end{proof}
We now note that any set $S\subseteq A$ that is an unsuccessful value query is also a poor approximation of the optimal contract: 

\begin{lemma} \label{lemma: unsuccessful approx.}
Let $S \subseteq A$. For $n \geq 64$, if a value query on $S$ is unsuccessful, then $g(S) \leq \frac{4g(G)}{\sqrt{m}}$.
\end{lemma}
\begin{proof}
    We start by noting that for any $S \subseteq A$ and $i\in S$, it holds that $f_G(i \mid S \setminus \{i\})\le \frac{1}{n}$, which implies $\frac{c_i}{f_G(i \mid S \setminus \{i\})} = \frac{1/(2mn)}{f_G(i \mid S \setminus \{i\})} \geq \frac{1}{2m}$.
    
    Now, if $|S| \geq 2m$, then $g(S) = \left(1-\sum_{i\in S}\frac{c_i}{f_G(i \mid S \setminus \{i\})}\right)f(S) \le \left(1-|S|\cdot \frac{1}{2m}\right)f(S) \le 0$, as needed.
    
    Otherwise, since $n\ge 64$, we have $|S| < 2m \le m^{\frac{3}{2}}$, and by our definition of an unsuccessful value query,  it holds that $|S \cap G| \le \sqrt{m}$, meaning $g(S) \le f_G(S) \le \frac{2\sqrt{m}}{n} = \frac{4g(G)}{\sqrt{m}}$, as needed.
\end{proof}

From Lemma~\ref{lemma: unsuccessful approx.}, if we assume without loss of generality that an algorithm value queries the set $S$ which it returns as a contract, we can say that an algorithm with no successful value queries achieves at most a $\frac{4}{\sqrt{m}}$-approximation. We therefore conclude the proof by showing that the probability of any algorithm that makes a polynomial number of value queries having a successful query is negligible.

\begin{lemma}
    For $n \ge 512$, if a determinstic algorithm makes at most $k$ value queries, the probability of at least one query being successful is at most $k \cdot e^{-\frac{\sqrt{m}}{4}}$.
\end{lemma}

\begin{proof}
    For any $m$ such that $k \cdot e^{-\frac{\sqrt{m}}{4}} \ge 1$ we are done. Otherwise, for $m$ such that $k \cdot e^{-\frac{\sqrt{m}}{4}} < 1$, 
    let $ALG$ be a deterministic algorithm that makes at most $k$ value queries. Our first step is to show that adaptivity doesn't  help $ALG$, which allows us to apply union bound. More precisely, we show the existence of non-adaptive queries $S_1, \dots, S_\ell$ such that $\ell \le k$ and \[
    \Pr_G[ALG \text{ makes a successful query}] \le \Pr_G[\text{one of }S_1, \dots, S_\ell\text{ is a successful query}].
    \]
    Let $S_i$ be the $i$-th query that $ALG$ asks after the answers to all previous queries $S_j$ were $\frac{1}{n} \max \left(\sqrt{m}, \frac{|S_j|}{\sqrt{m}}\right)$ for all $j<i$, and let $\ell\leq k$ be the index of the last query asked in this scenario\footnote{Since $k \cdot e^{-\frac{\sqrt{m}}{4}} < 1$, the sequence of sets $S_1,\ldots,S_\ell$ is well defined since by union bound and Lemma~\ref{lemma: successful negligible} there is a positive probability that all of $S_1, \dots, S_\ell$ are unsuccessful, in which case the answer to query $S_j$ is indeed $\frac{1}{n} \max \left(\sqrt{m}, \frac{|S_j|}{\sqrt{m}}\right)$.}.
    To see  that
    \[
    \Pr_G[ALG \text{ makes a successful query}] \le \Pr_G[\text{one of }S_1, \dots, S_\ell\text{ is a successful query}],
    \]
    let $T_1, \dots, T_{\ell'}$ be the (perhaps adaptive) queries $ALG$ actually makes. We will show that if at least one of those is successful, then at least one of $S_1, \dots, S_\ell$ is successful. Assume that one of $T_1, \dots, T_{\ell'}$ is successful, and let $i$ be the lowest index such that $T_i$ is successful. If $i=1$, note that since $ALG$ is deterministic, then $S_1 = T_1$, as needed. 
    Otherwise, by definition of $i$, for any $j < i$, $T_j$ is an unsuccessful query, then the answer of $ALG$ to value query $T_j$ is $\frac{1}{n} \max \left(\sqrt{m}, \frac{|T_j|}{\sqrt{m}}\right)$, which by induction gives us for any $j \le i$, $T_j = S_j$, meaning $S_i$ is a successful query, as needed.

    Now, by the union bound and Lemma~\ref{lemma: successful negligible} we have 
    \[
    Pr_G[ALG \text{ makes a successful query}] \le Pr_G[\text{one of }S_1, \dots, S_{\ell}\text{ is a successful query}] \le k \cdot e^{-\frac{\sqrt{m}}{4}}. \qedhere
    \]
\end{proof}
Note that for any $k$ that is polynomial in $n$ it holds that $k\cdot e^{-\frac{\sqrt{m}}{4}} = k\cdot e^{-\frac{n^{1/6}}{4}}$ is negligible, which concludes the proof.
\end{proof}

\section{Hardness of Approximation for Multi-Action, Submodular $f$} \label{sec: multi-action-submodular}
In this section, we use Proposition~\ref{prop:coverage hardness} to strengthen a hardness result of \citep{combinatorial-contracts}, showing that the optimal contract for multi-action settings with submodular $f$ is not only hard to compute exactly, but also to approximate within any constant.

Before presenting our result, we recall the multi-action model from Section~\ref{sec:model}.
The principal interacts with a single agent, who faces a set $A$ of $n$ actions, and can choose any subset $S \subseteq A$ of them. 
Every action $i \in A$ is associated with a cost $c_i \ge 0$, and the cost of a set $S$ of actions is $\sum_{i\in S} c_i$. The success probability function $f(S)$ denotes the probability of a project's success when the agent chooses the set of actions $S$. 
A contract is defined by a single parameter $\alpha \in [0,1]$, which denotes the payment to the agent upon the project's success. The principal's objective is to find a contract $\alpha$ that maximizes her utility $u_P(\alpha)=(1-\alpha)f(S_{\alpha})$, where $S_{\alpha} \in \arg\max_{S \subseteq A} u_A(\alpha, S) = \arg\max_{S \subseteq A} f(S)\cdot \alpha - \sum_{i\in S} c_i$ is the agent's best response to a contract $\alpha$, with tie-breaking in favor of the principal.

We are now ready to present the theorem.
\begin{theorem}\label{thm:multi-action hardness}
    In the multi-action model, %presented in \citep{combinatorial-contracts}, 
    for submodular (and even normalized unweighted coverage) success probability function $f$, no polynomial time algorithm with value oracle access can approximate the optimal contract within any constant factor, unless $P=NP$. 
\end{theorem}

\begin{proof}
Let $\beta \in \left(0,\frac{1}{12}\right)$, we prove that no poly-time $12\beta$-approximation algorithm exists, by reducing from the hardness presented in Proposition~\ref{prop:coverage hardness}.

Let $(k, f'=(U', A', h'))$ be the input to our reduction per Proposition~\ref{prop:coverage hardness} with $M=2, \varepsilon=\beta^4$. We build the following contract instance: 
\begin{itemize}
    \item The set of actions is $A = A'\cup \{0\}$.
    \item The success probability function is $f(S)=\frac{1}{2}(f'(S\cap A')+1[0 \in S])$. 
    \item The costs are $c_i = \frac{1-\beta^2}{2k}$ for all $i\in A'$ and $c_0=\frac{1}{2}(1-\beta^3)$.
\end{itemize}
\begin{claim}
    The success probability function defined above is a normalized unweighted coverage function. 
\end{claim}
\begin{proof}
    Let $U = U'
    \times \{0,1\}$, and define \[h(i)=\begin{cases}
        h'(i)\times \{0\} & i \in A' \\
        U'\times \{1\} & i = 0
    \end{cases}
    \]
    The normalized unweighted coverage function defined by $(U, A, h)$ is equal to $f$.
\end{proof}

\begin{lemma} \label{lemma: contracts separation}
In the contract problem instance defined above, the following holds:
\begin{enumerate}
    \item If $(k, f')$ satisfies the first condition of Proposition~\ref{prop:coverage hardness}, any $\alpha_0$ which is a $12\beta$-approximation of the optimal contract satisfies $\alpha_0 < 1-\beta^3$.
    \item If $(k, f')$ satisfies the second condition of Proposition~\ref{prop:coverage hardness}, any $\alpha_0$ which is a $12\beta$-approximation of the optimal contract satisfies $\alpha_0 \ge 1-\beta^3$.
\end{enumerate} 
\end{lemma}
Note that by proving Lemma~\ref{lemma: contracts separation} we conclude the proof of Theorem~\ref{thm:multi-action hardness}.

\begin{proof} [Proof of Lemma~\ref{lemma: contracts separation}]
We first observe that the agent's best response to a contract $\alpha < 1-\beta^2$ is to take no actions. This holds since for any $i\in A'$ it holds that 
\[
\alpha f(\{i\}) - c_i = \alpha\cdot \frac{1}{2} f'(\{i\}) - c_i = \alpha \cdot\frac{1}{2k} - \frac{1-\beta^2}{2k} < 0,
\]
where the second equality is since Proposition~\ref{prop:coverage hardness} guarantees $f'(\{i\}) = \frac{1}{k}$ for any $i\in A'$. It also holds that for $i=0$ we have
\[
\alpha f(\{i\}) - c_i = \alpha \cdot \frac{1}{2} - \frac{1}{2}
(1-\beta^3) < 0.
\]
This means the agent's utility from any non-empty set $S\subseteq A$ is strictly less than $0$, since
\[
u_A(\alpha, S) = \alpha f(S) - \sum_{i\in S} c_i \le \alpha \sum_{i\in S} f(\{i\}) - \sum_{i\in S} c_i = \sum_{i\in S} \alpha f(\{i\}) - c_i < 0,
\]
where the first inequality follows by subadditivity.

Note that the same arguments show that, given the contract $\alpha = 1-\beta^2$, the agent has non-positive utility from any non-empty set $S\subseteq A$.
\paragraph{Case 1: $(k, f')$ satisfies the first condition of Proposition~\ref{prop:coverage hardness}. } Let $S\subseteq A'$ be a set that satisfies the condition (i.e. $|S| = k$ and $f'(S)=1$). Under the contract $\alpha=1-\beta^2$, the agent's utility from the set $S$ is
\[
u_A(\alpha, S) = \alpha f(S) - \sum_{i\in S} c_i = (1-\beta^2)  \frac{1}{2} - |S|\frac{1-\beta^2}{2k} = 0.
\]
Since, as we noted earlier, no set has a greater utility to the agent, and ties are broken in favor of the principal, this implies that the agent's best response $S_\alpha$ satisfies $f(S_\alpha) \ge f(S) \ge \frac{1}{2}$. It follows that the principal's utility from the contract $\alpha = 1-\beta^2$ satisfies
\[
u_P(\alpha) = f(S_\alpha) (1-\alpha) \ge \frac{1}{2} \beta^2.
\]
Now, let $\alpha_0$ be some $12\beta$-approximation of the optimal contract. This implies that
\[
u_P(\alpha_0) \ge 12\beta \cdot u_P(1-\beta^2) \ge 12\beta \cdot \frac{1}{2}\beta^2 > \beta^3.
\]
On the other hand, trivially it holds that $u_P(\alpha_0) \le 1-\alpha_0$, which gives us $\alpha_0 < 1-\beta^3$, as needed.

\paragraph{Case 2: $(k, f')$ satisfies the second condition of Proposition 3.2.} We start by arguing that the principal's utility from the contract $\alpha=1-\beta^3$ is at least $\frac{1}{2}\beta^3$. First, we show that the agent's best response $S_\alpha$ will always include the action $0$. Indeed, if we assume by contradiction that $0\notin S_\alpha$, by adding $0$ to $S_\alpha$ we do not change the agent's utility:
\[
u_A(S_\alpha \cup \{0\}) - u_A(S_\alpha) = \alpha f(0 \mid S_\alpha) - c_0 = (1-\beta^3) \cdot \frac{1}{2} - \frac{1}{2}(1-\beta^3)=0.
\]
This means that $S_\alpha \cup \{0\}$ has the same utility to the agent, but a greater utility to the principal, contradicting tie-breaking in favor of the principal.
This implies that
\[
u_P(\alpha) = f(S_\alpha)(1-\alpha) \ge \frac{1}{2}(1-\alpha) = \frac{1}{2}\beta^3.
\]
Now, let $\alpha < 1-\beta^3$ be some contract, we show that $u_P(\alpha)< 6 \beta^4 \le 12\beta \cdot u_P(1-\beta^3)$, completing the proof of Lemma~\ref{lemma: contracts separation}. If $\alpha < 1-\beta^2$, as argued before, the agent's best response is the empty set, which means the principal's utility is $0$ and we are done. Otherwise, if $1-\beta^2 \le \alpha < 1-\beta^3$, since $S_\alpha$ is the agent's best response, it is clear that $0 \notin S_\alpha$ (since otherwise $S_\alpha \setminus \{0\}$ has a strictly better utility to the agent). From this we conclude that $f(S_\alpha) = \frac{1}{2}f'(S_\alpha) \le \frac{1}{2}$. Since $S_\alpha$ must have a non-negative utility to the agent, it also holds that
\begin{equation} \label{ineq: size of BR}
|S_\alpha| \frac{1-\beta^2}{2k} = \sum_{i\in S_\alpha} c_i \le \alpha f(S_\alpha) \le f(S_\alpha).
\end{equation}
Since $f(S_\alpha) \le \frac{1}{2}$, Inequality~\ref{ineq: size of BR} implies that $|S_\alpha| \le \frac{k}{1-\beta^2} \le M \cdot k$. This allows us to use condition (2) of Proposition~\ref{prop:coverage hardness}, which implies that $f(S_\alpha) \le \frac{1}{2}(1-e^{-|S_\alpha|/k}+\varepsilon)$. Denoting $x=\frac{|S_\alpha|}{k}$, and plugging this into Inequality~\ref{ineq: size of BR}, we get the inequality
\[
x \frac{1-\beta^2}{2} \le \frac{1}{2} (1-e^{-x} + \varepsilon) \le \frac{1}{2} \left(x-\frac{1}{4}x^2 + \varepsilon\right),
\]
where the last inequality is since $e^{-x} \ge 1-x+\frac{1}{4}x^2$ for any $x\in [0,2]$. By rearranging, we get that 
\[
\frac{1}{4} x^2 -\beta^2 x- \varepsilon \le 0,
\]
implying 
\[
x \le \frac{\beta^2 + \sqrt{\beta^4 + \varepsilon}}{1/2} \le \frac{\beta^2 + \sqrt{2} \beta^2}{1/2} < 6\beta^2 
\]
This means the principal's utility is at most
\[
u_P(\alpha) = f(S_\alpha) (1-\alpha) = \frac{1}{2}f'(S_\alpha) (1-\alpha) \leq |S_\alpha| \cdot \frac{1}{2k} (1-\alpha) < 6\beta^2(1-\alpha) < 6\beta^4,
\]
as needed.
\end{proof}
This concludes the proof of Theorem~\ref{thm:multi-action hardness}
\end{proof}

\section{Hardness of Approximation for Multi-Action, XOS $f$}

In this section, we show a hardness of approximation result for the multi-action model with XOS success probability functions.
More formally, we prove the following theorem:
\begin{theorem}\label{thm:multi-action hardness XOS}
    In the multi-action model, % presented in \citep{combinatorial-contracts}, 
    for XOS $f$, for any $\varepsilon > 0$, no polynomial time algorithm with value query access can approximate the optimal contract to within a factor of $n^{-\frac{1}{2}+\varepsilon}$, unless P=NP. 
\end{theorem}

Our proof of Theorem~\ref{thm:multi-action hardness XOS}  relies on the hardness of approximating $\omega(G)$, which is the size of the largest clique in the graph $G$. This hardness result was presented by \citet{Hastad_1999, zuckerman2006}, who prove the following theorem:

\begin{theorem} [\citep{Hastad_1999, zuckerman2006}]
    For all $\varepsilon > 0$, it is NP-hard to approximate $\omega(G)$ to within $n^{-1+\varepsilon}$. 
\end{theorem}

Our technique is to use any $\beta$-approximation algorithm of the optimal contract to distinguish between the cases $\omega(G) \le \delta$ and $\omega(G) \ge \frac{2\delta}{\beta^2}$, for any $\delta>0$. 
Solving this promise problem allows us to get a guarantee of either $\omega(G) > \delta$ or $\omega(G) < \frac{2\delta}{\beta^2}$. By iterating over $\delta_i = 2^i$, we can get some $i$ for which we are guaranteed $\omega(G) > \delta_i$ and $\omega(G) < \frac{2\delta_{i+1}}{\beta^2}$. This allows us to approximate $\omega(G)$ within a factor of $\frac{\beta^2}{4}$ (see formal arguments in Lemma~\ref{lemma: approximating clique}), which implies Theorem~\ref{thm:multi-action hardness XOS} from \citep{Hastad_1999, zuckerman2006}.

In Section~\ref{subsec:distinguishing_clique} we show how to use a $\beta$-approximation algorithm for the contract problem to distinguish between the two cases $\omega(G) \le \delta$ and $\omega(G) \ge \frac{2\delta}{\beta^2} $ in polynomial time for any $\delta > 0$, and in Section~\ref{subsec:approximating_clique} we formally prove that this distinction gives us the ability to approximate $\omega(G)$ to within a factor of $\frac{\beta^2}{4}$ in polynomial time, thus concluding the proof of Theorem~\ref{thm:multi-action hardness XOS}.

\subsection{Distinguishing Between $\omega(G) \le \delta$ and $\omega(G) \ge \frac{2\delta}{\beta^2}$} \label{subsec:distinguishing_clique}
In this section, we prove the following lemma:
\begin{lemma} \label{lemma: distinguishing clique}
    Algorithm~\ref{alg:distinguishing_clique} runs in polynomial time, given oracle access to a $\beta$-approximation of the optimal contract for XOS functions.  Additionally, on input $(G, \delta)$ composed of a graph $G$ and a positive integer $\delta$
    it holds that:
    \begin{enumerate}
        \item If $\omega(G) \le \delta$, then Algorithm~\ref{alg:distinguishing_clique} returns $\smallClique$.
        \item If $\omega(G) \ge \frac{2\delta}{\beta^2} $, then Algorithm~\ref{alg:distinguishing_clique} returns $\largeClique$.
        \item If $\delta < \omega(G) < \frac{2\delta}{\beta^2}$, then Algorithm~\ref{alg:distinguishing_clique} returns either $\smallClique$ or \largeClique.
    \end{enumerate}
\end{lemma}

We note that in Algorithm~\ref{alg:distinguishing_clique} we build a contract problem instance with a success probability function that attains values greater than $1$. This is done for simplicity, and the result clearly holds for success probability functions that attain values within $[0,1]$, by normalizing both $f$ and the costs $c_i$ by $f(V')$ (the maximum value of $f$).
\begin{algorithm}
    \caption{Distinguishing Between $\omega(G) \le \delta$ and $\omega(G) \ge \frac{2\delta}{\beta^2}$} \label{alg:distinguishing_clique}
    \begin{algorithmic}[1]
            \State Given a graph $G=(V,E)$ and $\delta \in \mathbb{N}^+$, build the graph $G'=(V',E')$ such that \[
            \begin{split}
                V' &= V \cup [\delta] \\
                E' &= E \cup \{\{i, j\} \mid i,j\in [\delta] \land i\ne j\}
            \end{split}
            \]
            \State Denote $\varepsilon = \frac{2}{\beta}-1$, $M=|V'|+\varepsilon$
            \State Get $\alpha_0$ which is a $\beta$-approximation to the optimal contract in the instance $(V', f: 2^{V'} \rightarrow \mathbb{R}_{\ge 0}, \{c_i=M\}_{i\in V'})$ where\[
            f(S) = (M+\mathbb{1}[S \text{ is a clique in $G'$}])\cdot|S| + \min (|S|, \delta) \cdot \varepsilon 
            \] 
            \State If $\alpha_0<\frac{M}{M+1}$, return \smallClique, otherwise return \largeClique.
    \end{algorithmic} 
\end{algorithm}

First, we note the following lemma, which proves our oracle call in step 2 is valid:
\begin{lemma}
    The function $f$ as defined in step 3 of Algorithm~\ref{alg:distinguishing_clique} is monotone,  XOS, and value queries can be computed in polynomial time.
\end{lemma}
\begin{proof}
    For any non-empty subset $T\subseteq V'$, we define the additive function $f_T: V'\rightarrow \mathbb{R}_{\ge 0}$ according to the weights
    \[
    a_i^T = \begin{cases}
        M+\mathbb{1}[T\text{ is a clique}] + \varepsilon\cdot \frac{\min(|T|, \delta)}{|T|} & i\in T \\
        0 & else,
    \end{cases}
    \]
    where $f_T(S) = \sum_{i\in S } a_i^T$.
    
    Let $S\subseteq V'$, we claim that
    \[
    f(S) = \max_{T\in 2^{V'}\setminus \{\emptyset\}} f_T(S),
    \]
    which proves both XOS and monotonicity since all weights $a_i^T$ are non-negative. 
    First, by taking $T=S$, it is clear that 
    \[
    \max_{T\in 2^{V'}\setminus \{\emptyset\}} f_T(S) \ge f_S(S) = f(S).
    \]

    Let $T\subseteq V'$. We will show that $f_T(S) \le f(S)$, thus concluding the proof.
    
    If $S \subseteq T$, it is clear that for any $i\in S$ we have $a_i^T \le a_i^S$ which implies $f_T(S) \le f_S(S) = f(S)$.
    
    Otherwise, if $S \not \subseteq T$, it is clear that $|S \cap T| < |S|$, which gives us
    \[
    \begin{split}
    f_T(S) &= \sum_{i\in S} a_i^T \le |S \cap T| \left(M+1 + \varepsilon \frac{\min(|T|, \delta)}{|T|}\right) \le |S\cap T| (M+1) + \varepsilon \min (|S|, \delta) \\
    &< (|S \cap T| + 1) M + \varepsilon \min (|S|, \delta) \le |S| \cdot M + \varepsilon \min(|S|, \delta) \le f(S),
    \end{split}
    \]
    where the strict inequality is since $M > |S \cap T|$.

    A value oracle of $f$ can be computed efficiently by checking, for a given set $S$, whether it is a clique or not.
    \end{proof}

Next, we prove the correctness of Algorithm~\ref{alg:distinguishing_clique}. We start by characterizing the agent's best response.
\begin{lemma} \label{lemma:agent's best response}
    On any input $(G, \delta)$, when considering the contract problem instance defined in line 3 of Algorithm~\ref{alg:distinguishing_clique}:
    \begin{enumerate}
        \item The agent's best response to a contract $0 < \alpha < \frac{M}{M+1+\varepsilon}$ is $\emptyset$.
        \item The agent's best response to a contract $\frac{M}{M+1+\varepsilon} \le \alpha < \frac{M}{M+1}$ is a clique of size $\delta$.
        \item The agent's best response to a contract $ \frac{M}{M+1} \le \alpha < 1$ is a maximum size clique.
    \end{enumerate}
\end{lemma}
\begin{proof}
    Let $\alpha\in (0,1)$.
    \paragraph{If $\alpha < \frac{M}{M+1+\varepsilon}$.}
    The agent's utility from a set $S\subseteq V'$ is 
    \[
    u_A(\alpha, S) = f(S)\cdot \alpha - |S|\cdot M \le (M+1+\varepsilon) |S| \cdot \alpha - |S| \cdot M,
    \]
    which is strictly negative unless $|S| = 0$.
    \paragraph{If $\frac{M}{M+1+\varepsilon} \le \alpha < \frac{M}{M+1}$.}
    Let $T$ be a clique of size $\delta$ in $G'$ (such a clique exists since we can simply take the $\delta$ vertices we added to $G$). The agent's utility from $T$ is 
    \[
    u_A(\alpha, T) = f(T) \cdot \alpha - |T| \cdot M = (M+1+\varepsilon) \cdot \delta \cdot \alpha - \delta \cdot M = \delta \left((M+1+\varepsilon) \cdot \alpha - M \right). 
    \]
    First, we note that any set $S$ with $|S| > \delta$ is strictly worse than $T$:
    \[
    \begin{split}
    u_A(\alpha, S) - u_A(\alpha, T) &\le (|S|(M+1)+\varepsilon\delta )\cdot \alpha - M |S| - \delta \left((M+1+\varepsilon) \cdot \alpha - M \right) \\
    &= (|S|-\delta)\left( (M+1) \alpha -M\right) < 0.
    \end{split}
    \]
    Additionally, any $S$ such that $|S| \le \delta$ is not better than $T$:
    \[
    u_A(\alpha, S) \le |S| \left((M+1+\varepsilon)\cdot \alpha - M\right) \le u_A(\alpha, T).
    \]
    These two facts, together with the fact that out of the sets with $|S| \le \delta$, $T$ maximizes $f(T)$ and the agent breaks ties in favor of the principal, give us that the agent's best response is $T$ (or another clique of size $\delta$).
    
    \paragraph{If $ \frac{M}{M+1} \le \alpha <1$.} Let $T$ be a maximum size clique. The agent's utility from $T$ is
    \[
    u_A(\alpha, T) = f(T) \cdot \alpha - |T| \cdot M = \alpha((M+1)|T| + \varepsilon \delta) - |T| \cdot M \ge \alpha\cdot \varepsilon \cdot \delta.
    \]
    Again, we note that any set $S$ with $|S| > |T|$ is strictly worse than $T$. Since $S$ cannot be a clique, we get that:
    \[
    \begin{split}
    u_A(\alpha, S) = \alpha(M|S| + \varepsilon \delta) - |S| \cdot M < \alpha \cdot \varepsilon \cdot \delta 
    \end{split}
    \]
    Likewise, any set $S$ with $|S| \le |T|$ isn't better than $T$:
    \[
    \begin{split}
    u_A(\alpha, S) &\le (|S|(M+1)+\varepsilon \delta) \cdot \alpha - |S| \cdot M = \varepsilon \cdot \delta \cdot \alpha + |S|((M+1)\cdot \alpha - M) \\
    &\le \varepsilon \cdot \delta \cdot \alpha + |T|((M+1)\cdot \alpha - M) = u_A(\alpha, T).
    \end{split}
    \]
    Again, since $T$ maximizes $f(T)$ (among all sets of size at most $|T|$) and the agent breaks ties in favor of the principal, the the agent's best response is $T$ (or another clique of maximum size).
\end{proof}

A direct corollary of Lemma~\ref{lemma:agent's best response} is the following lemma, which is a full characterization of the principal's utility from any contract: 
\begin{corollary} \label{cor: principal's utility} 
    On any input $(G, \delta)$, when considering the contract problem instance defined in line 3 of Algorithm~\ref{alg:distinguishing_clique}:
    \begin{enumerate}
        \item The principal's utility from a contract $0 < \alpha < \frac{M}{M+1+\varepsilon}$ is 
        \[
        u_P(\alpha) = 0.
        \]
        \item The principal's utility from a contract $\frac{M}{M+1+\varepsilon} \le \alpha < \frac{M}{M+1}$ is \[
        u_P(\alpha) = \left((M+1)\cdot \delta + \min (\delta, \delta) \cdot \varepsilon \right)\cdot (1-\alpha) = (M+1+\varepsilon)\delta (1-\alpha).
        \]
        \item The principal's utility from a contract $\frac{M}{M+1} \le \alpha < 1$ is
        \[
        u_P(\alpha) = \left((M+1)\omega(G') +\delta \varepsilon \right)(1-\alpha).
        \]
    \end{enumerate}
\end{corollary}

We are now ready to prove Lemma~\ref{lemma: distinguishing clique}
\begin{proof} [Proof of Lemma~\ref{lemma: distinguishing clique}]
    Let $(G,\delta)$ be the input of Algorithm~\ref{alg:distinguishing_clique}. To prove Lemma~\ref{lemma: distinguishing clique} it suffices to prove that if $w(G) \le \delta$, Algorithm~\ref{alg:distinguishing_clique} returns $\smallClique$ and if $w(G) \ge \frac{2\delta}{\beta^2}$, it returns $\largeClique$.

    \paragraph{If $w(G) \le \delta$.} In this case,  $\omega(G') = \delta$. Assume towards contradiction that Algorithm~\ref{alg:distinguishing_clique} returns $\largeClique$, i.e., $\alpha_0 \ge \frac{M}{M+1}$. Since $\alpha_0$ is a $\beta$-approximation, we have 
    \[
    u_P(\alpha_0) \ge \beta \cdot u_P\left(\frac{M}{M+1+\varepsilon}\right).
    \]
    Corollary~\ref{cor: principal's utility} gives us
    \[
    u_P\left(\frac{M}{M+1}\right) \ge u_P(\alpha_0).
    \]
    Putting the two inequalities together, and substituting them with the expressions for the principal's utility from Corollary~\ref{cor: principal's utility} gives us
    \begin{equation} \label{eq: first utility ineq.}
    \left((M+1)\omega(G') +\delta \varepsilon \right)\left(1-\frac{M}{M+1}\right) \ge \beta \cdot (M+1+\varepsilon)\delta \left(1-\frac{M}{M+1+\varepsilon}\right).
    \end{equation}
    Note that since $\varepsilon = \frac{2}{\beta}-1$, then
    \[
    \left(1-\frac{M}{M+1+\varepsilon}\right)=  \frac{1+\varepsilon}{M+1+\varepsilon} = \frac{2}{\beta} \cdot \frac{1}{M+1+\varepsilon} > \frac{2}{\beta} \cdot \frac{1}{2(M+1)} = \frac{1}{\beta} \cdot \left(1-\frac{M}{M+1}\right),
    \]
    where the inequality is since $M \ge \varepsilon$. Plugging this into Inequality~\ref{eq: first utility ineq.} and dividing both sides by $\left(1-\frac{M}{M+1}\right)$ gives us 
    \[
    (M+1)\omega(G') +\delta \varepsilon   >  (M+1+\varepsilon)\delta.
    \]
    This gives us $\omega(G') > \delta$, contradiction.

    \paragraph{If $w(G) \ge \frac{2\delta}{\beta^2}$.}  In this case, $\omega(G') \ge \frac{2\delta}{\beta^2}$. Assume by contradiction that Algorithm~\ref{alg:distinguishing_clique} returns $\smallClique$, i.e., $\alpha_0 < \frac{M}{M+1}$. Since $\alpha_0$ is a $\beta$-approximation, we have 
    \[
    u_P(\alpha_0) \ge \beta \cdot u_P\left(\frac{M}{M+1}\right).
    \]
    Corollary~\ref{cor: principal's utility} gives us
    \[
    u_P\left(\frac{M}{M+1+\varepsilon}\right) \ge u_P(\alpha_0).
    \]
    Putting the two inequalities together, and substituting them with the expressions for the principal's utility from Corollary~\ref{cor: principal's utility} gives us
    \begin{equation} \label{eq: second utility ineq.}
     (M+1+\varepsilon)\delta \left(1-\frac{M}{M+1+\varepsilon}\right) \ge \beta \cdot \left((M+1)\omega(G') +\delta \varepsilon \right)\left(1-\frac{M}{M+1}\right).
    \end{equation}
    Note that since $\varepsilon = \frac{2}{\beta}-1$, then
    \[
    \left(1-\frac{M}{M+1+\varepsilon}\right)=  \frac{1+\varepsilon}{M+1+\varepsilon} = \frac{2}{\beta} \cdot \frac{1}{M+1+\varepsilon} < \frac{2}{\beta} \cdot \frac{1}{M+1}.
    \]
    Plugging this into the LHS of Inequality~\ref{eq: second utility ineq.} and dividing both sides by $\frac{2}{\beta} \cdot \left(1-\frac{M}{M+1}\right)$ gives us 
    \[
    (M+1+\varepsilon)\delta > \frac{\beta^2}{2}  \cdot \left((M+1)\omega(G') +\delta \varepsilon \right) \ge \frac{\beta^2}{2}  \cdot (M+1+\varepsilon)\omega(G')
    \]
    where the second inequality is from $\omega(G') \ge \delta$. This gives us $\omega(G') < \frac{2\delta}{\beta^2}$, contradiction.
\end{proof}
\subsection{Using the Differentiation to Approximate $\omega(G)$} \label{subsec:approximating_clique}
In this section we show how to use the guarantees given by Section~\ref{subsec:distinguishing_clique} to get a $\beta^2/4$-approximation algorithm for $\omega(G)$, given an oracle access to a $\beta$-approximation of the optimal contract, thus concluding the proof of Theorem~\ref{thm:multi-action hardness XOS}.

\begin{lemma} \label{lemma: approximating clique}
    Algorithm~\ref{alg:approx_clique}, given oracle access to a $\beta$-approximation of the optimal contract for XOS functions, runs in polynomial time, and on input $G=(V,E)$ gives a $\frac{\beta^2}{4}$-approximation of $\omega(G)$. 
\end{lemma}

\begin{algorithm}
    \caption{Approximation $\omega(G)$} \label{alg:approx_clique}
    \begin{algorithmic}[1]
            \State Given a graph $G=(V,E)$
            \For{$i \gets 0$ to $\lfloor\log_2(|V|)\rfloor$}
                \State Run Algorithm~\ref{alg:distinguishing_clique} on $(G, \delta=2^i$), and denote its answer by $a(i)$.
            \EndFor
            \If {$a(0) = \smallClique$}
                \State \Return 1.
            \Else
                \State \Return $2^{i_{\max}}$ where $i_{\max}$ is the maximal $i$ such that $a(i)=\largeClique$.
            \EndIf
    \end{algorithmic} 
\end{algorithm}

\begin{proof} [proof of Lemma~\ref{lemma: approximating clique}]
    If Algorithm~\ref{alg:approx_clique} returns $1$ on line 6, then from Lemma~\ref{lemma: distinguishing clique}, we know that $\omega(G) \le \frac{2}{\beta^2}$, meaning the algorithm's output is a $\frac{\beta^2}{2}$-approximation, as needed.
    
    Otherwise, let $i_{\max}$ be the maximal $i$ such that $a(i)=\largeClique$. By Lemma~\ref{lemma: distinguishing clique} we know that $\omega(G) \ge 2^{i_{\max}}$. 
    
    If $i_{\max}=\lfloor \log_2 (|V|) \rfloor$, returning $2^{i_{\max}} \ge \frac{|V|}{2}$ gives us a $\frac{1}{2}$-approximation (since trivially $\omega(G) \le |V|$), as needed.
    
    Finally, if $i_{\max}<\lfloor \log_2 (|V|) \rfloor$, by our choice of $i_{\max}$ we know that $a(i_{\max}+1)=\smallClique$, meaning from Lemma~\ref{lemma: distinguishing clique} that $\omega(G) \le \frac{2\cdot 2^{i+1}}{\beta^2}$, which means returning $2^i$ gives us a $\frac{\beta^2}{4}$-approximation, as needed.
\end{proof}

\bibliographystyle{abbrvnat}

\bibliography{bib.bib}
\appendix
\section{Proof of Proposition~\ref{prop:coverage hardness}} \label{app: coverage promise problem proof}
\subsection{The MAX 3SAT-5 Problem} \label{subsec:max-3sat-5}
In this section we present the MAX 3SAT-5 problem.

A 3CNF-5 formula is a CNF formula $\phi = (n, \{X_1, \dots, X_n\}, \{\phi_1, \dots, \phi_{\frac{5n}{3}}\})$ with $n$ variables $X_1, \dots, X_n$ and $\frac{5n}{3}$ clauses $\phi_1, \dots, \phi_{\frac{5n}{3}}$, in which every clause contains exactly three literals, every variable appears in exactly five clauses, and a variable does not appear in a clause more than once.

The MAX 3SAT-5 problem asks to find the maximum number of clauses that can be satisfied simultaneously in a 3CNF-5 formula.

\cite{feige} establishes the following proposition (stated as 
Proposition 2.1.2 in \cite{feige}).
\begin{proposition}[\cite{feige}] \label{prop: 3sat-5 hardness}
For some $\varepsilon > 0$, it is NP-hard to distinguish between satisfiable 3CNF-5 formulas and 3CNF-5 formulas in which at most a $(1-\varepsilon)$-fraction of the clauses can be satisfied simultaneously.
\end{proposition}
We refer to the problem defined in this proposition as the 3CNF-5 separation problem.

\subsection{The $k$-Prover Proof System} \label{subsec:k-prover system}
The $k$-prover proof system presented by \cite{feige} has a parameter $\ell$ that affects its runtime and error rate. In it, a verifier interacts with $k$ provers, denoted $P_1, \dots, P_k$.
Consider a binary code that contains $k$ code words $w_1, \dots, w_k \in \{0,1\}^\ell$, each of length $\ell$, and
weight $\ell/2$, and Hamming distance at least $\ell/3$ between any two different code words.

In the $k$-prover system, the verifier performs the following steps when presented with a 3CNF-5 formula $\phi = (n, \{X_1, \dots, X_n\}, \{\phi_1, \dots, \phi_{\frac{5n}{3}}\})$:
\begin{enumerate}
    \item The verifier initiates by uniformly sampling a random state $r$ from the set of possible random states $R = [\frac{5n}{3}]^\ell\times[3]^\ell$.
    \item Based on the random first $\ell$ coordinates of state $r$, the verifier chooses clauses $C_i = \phi_{r_i}$. Using the last $\ell$ coordinates of $r$, from each clause $C_i$, a variable denoted as $x_i$ is selected using $r_{\ell+i}$. Specifically, $x_i$ represents the $r_{\ell+i}$-th variable appearing in $C_i$. The variables $x_1, \dots, x_\ell$ are referred to as the \emph{distinguished variables}, and may include repetitions. 
    \item The verifier has a set of possible questions $Q = [\frac{5n}{3}]^{\ell/2} \times [n]^{\ell/2}$ that can be directed to each of the provers which represents all possible choices of $\ell/2$ clauses and $\ell/2$ variables (with repetition). For every $i$, question $q_i(\phi, r)\in Q$, which consists of $\ell/2$ \textit{variable queries} and $\ell/2$ \textit{clause queries}, is sent to prover $P_i$.  
    To prover $P_i$, for each $j\in [\ell]$, the verifier queries either clause $C_j$ --- if $\left(w_i\right)_j = 1$ --- or the variable $x_j$ --- if $\left(w_i\right)_j = 0$. We call this the $j$th query, despite it not necessarily being the $j$th element in $q_i(\phi, r)$. The first $\ell/2$ elements of $q_i(\phi, r)$ represent the indices of the clauses to be queried (i.e., $\ell/2$ of the first $\ell$ elements of $r$). The last $\ell/2$ elements of $q_i(\phi, r)$ represent the indices of the variables to be queried. 
    \item The verifier receives an answer $a_i \in \{0, 1\}^{2\ell}$ from each prover $P_i$ that depends only on $\phi$ and on $q_i(\phi,r)$.\footnote{The prover $P_i$ is not aware of the random choice $r$ of the verifier, neither of the questions $q_{i'}(\phi,r)$ sent to prover $P_{i'}$ for  $i'\neq i$.} The answer $a_i$ consists of $2\ell$ bits, where each variable query is answered with an assignment to that variable, and each clause query is answered with an assignment to all three variables within the clause. Note that conflicting answers may be received for queries directed to the same prover. It is assumed, without loss of generality, that each answer to a clause query represents a satisfying assignment for that clause. In case the answer does not satisfy the clause, the verifier can simply flip the first bit of the three bits corresponding to that clause query's answer. 
    \item From each answer $a_i$, an assignment $\rho_i$ of the distinguished variables is extracted. The assignment assigns a value to the distinguished variable $x_j$ based on the answer to the $j$th query. More specifically, if $(w_i)_j = 0$, then $\rho_i$ simply assigns to $x_j$ the answer to the $j$th query (which directly queried $x_j$). If $(w_i)_j = 1$, then $\rho_i$ assigns to $x_j$ the bit corresponding to it from the 3-bit answer to the $j$th query. Note that if the distinguished variables contain repetitions, the extracted assignment may not be valid for the formula's variables corresponding to the distinguished variables; we allow such contradictions.
    \item The verifier \emph{weakly accepts} if \emph{at least one pair} of provers $P_i \ne P_{i'}$ agrees on the assignment to the distinguished variables (i.e., $\rho_i = \rho_{i'}$).
    \item The verifier \emph{strongly accepts} if \emph{every} pair of provers $P_i \ne P_{i'}$ agrees on the assignment to the distinguished variables.
\end{enumerate}

This k-prover proof system satisfies the following lemma (lemma 2.3.1 from \cite{feige}):

\begin{lemma}[\cite{feige}] \label{k-prover lemma}
Consider the $k$-prover proof system defined above and a 3CNF-5
formula $\phi$. If $\phi$ is satisfiable, then the provers have a strategy that causes the verifier to always strongly accept. If at most a $(1-\varepsilon)$-fraction of the clauses in $\phi$ are simultaneously satisfiable, then for every strategies of the provers the verifier weakly accepts with probability at most $k^2 \cdot 2^{-c\ell}$, where $c$ is a constant that depends only on $\varepsilon$.
\end{lemma}

\subsection{Construction of Our Coverage Function} \label{subsec:construction of coverage function}
In this section, we use the $k$-prover proof system to build a reduction from 3CNF-5 separation problem
{(defined in Proposition~\ref{prop: 3sat-5 hardness})}, thereby proving Proposition~\ref{prop:coverage hardness}. Specifically, we prove the following lemma:
\begin{lemma} \label{lemma: coverage reduction}
    For any $\varepsilon'$ that satisfies Proposition~\ref{prop: 3sat-5 hardness}, and for any $M > 1,~ 0 < \varepsilon < e^{-1}$ there exists a polynomial time algorithm that, 
    given a 3CNF-5 formula $\phi = (n, \{X_1, \dots, X_n\}, \{\phi_1, \dots, \phi_{\frac{5n}{3}}\})$, 
    returns a parameter $k'$ and a description of a normalized unweighted coverage function $f$ defined by $ (U, A, h)$, such that $k'$ is polynomial in $n$, and which satisfy the following conditions: 
    \begin{enumerate}
        \item For every $i\in A$, it holds that $f(\{i\}) = \frac{1}{k'}$.
        \item If $\phi$ is satisfiable, then there exists a set $S\subseteq A$ of size $k'$ s.t $f(S) = 1$.
        \item If at most a $(1-\varepsilon')$-fraction of the clauses of $\phi$ are simultaneously satisfiable, then for any $\beta < M$, any set $S\subseteq A$ of size $\beta k'$ has $f(S) \le 1-e^{-\beta} +\varepsilon$. 
    \end{enumerate}
\end{lemma}
Combining  Proposition~\ref{prop: 3sat-5 hardness}, and Lemma~\ref{lemma: coverage reduction} implies immediately Proposition~\ref{prop:coverage hardness}.

We start by describing the construction of $k'$ and $f=(U, A, h)$.

\paragraph{The construction of $k'$ and the coverage function $f$.} 
 Let $k, \ell \in \mathbb{N}$ be constants to be chosen later (they depend on $M, \varepsilon$). $k, \ell$ will be the parameters of the $k$-prover system which is the basis for our reduction. 
Let $\phi = (n, \{X_1, \dots, X_n\}, \{\phi_1, \dots, \phi_{\frac{5n}{3}}\})$ be a 3CNF-5 formula.
Consider the verifier in the $k$-prover proof system with parameter $\ell$ on input $\phi$. Let $R$ be the set of possible random states for the verifier (as in step 1), and let $Q$ be the set of possible questions (as in step 3). 

The algorithm returns $k'$ and $f$. $k'$ is defined as $k'=k\cdot |Q|$, which is polynomial in $n$ since $k,\ell$ are constant and $|Q| = \left(\frac{5n}{3}\right)^{\ell/2}n^{\ell/2}$. We next define $f$.  Let $L=2^\ell$ be the number of possible assignments of the distinguished variables. We use assignments and numbers $j\in [L]$  interchangeably (to switch between them we can think of a binary representation of $j-1$).  In what follows we specify the $U, A, h$.

We set
\[
U = [k]^L \times R.
\]
Since $k,\ell$ are constants, 
and $|R| = \left(\frac{5n}{3}\right)^\ell\cdot 3^\ell$ is polynomial in $n$, $|U|$ is polynomial in $n$. 
We next define the set of items $A$. We set
\[
A = \{ (q, a, i) \in Q \times \{0,1\}^{2\ell} \times [k] \mid a \text{ is a valid answer to question } q \}.
\]
Note that $|A| \le |Q| \cdot 2^{2\ell}\cdot k$  is polynomial in $n$. Intuitively, the item $(q, a, i)$ represents the case that prover $P_i$ answers question $q$ with answer $a$. The condition that $a$ is a valid answer to question $q$ depends on $\phi$ and means that all clause queries were answered with a satisfying assignment (which is assumed without loss of generality in the $k$-prover proof system definition).
Before describing $h$, we first specify our building blocks for $h$. For $r\in R, ~j\in [L], ~i\in[k]$ we define
\[
B(r, j, i) = \{u\in U \mid u_j = i \land u_{L+1} = r \}.
\]
Intuitively, $B(r, j, i)$ represents the case where under random state $r$ of the verifier, the assignment to the distinguished variables induced by $P_i$'s answer is $j$ (i.e., $\rho_i = j$). For any $r\in R$, we also denote $U_r = \{u\in U \mid u_{L+1} = r\}$, intuitively this is the case where the verifier selects randomness $r$. The sets $B(r, j, i)$ have immediate desirable properties, which we describe in the following two claims:
\begin{claim} \label{claim:covering U_r}
    For any $r\in R, j \in [L]$ it holds that  $\bigcup_{i=1}^k B(r, j, i) = U_r$. 
\end{claim}
\begin{proof}
    Let $r\in R, j\in [L]$, it holds that:
    \[
    \bigcup_{i=1}^k B(r, j, i) = \bigcup_{i=1}^k \{u\in U \mid u_j = i \land u_{L+1} = r \} = \{u\in U \mid u_{L+1} = r \} = U_r,
    \]
    where the second equality is since $u_j\in [k]$ for all $j\in [L]$. 
\end{proof}
\begin{claim} \label{claim: bounding disjoint cover}
    For any $r\in R,~ \mathcal{I} \subseteq [k] \times [L]$, if no two pairs in $\mathcal{I}$ share $j\in [L]$ (i.e., for every $j
    \in [L]$  it holds that $|\{i \mid (i,j)
     \in \mathcal{I}
     \}| \leq 1$),  then 
    \[
    \left|\bigcup_{(i,j) \in \mathcal{I}} B(r, j, i))\right| = \left(1-\left(1-\frac{1}{k}\right)^{|\mathcal{I}|}\right)|U_r|.
    \]
\end{claim}
\begin{proof}
    Let $r\in R$ and let $\mathcal{I} \subseteq [k] \times [L]$ be such that no two pairs in $\mathcal{I}$ share $j\in [L]$.
    We note that for every non-empty $\mathcal{I'}\subseteq \mathcal{I}$, since no two pairs $(i_1, j_1) \ne (i_2, j_2) \in \mathcal{I'}$ can have $j_1 = j_2$ we get
    \[
    \left|\bigcap_{(i,j) \in \mathcal{I'}} B(r, j, i)\right| = k^{L-|\mathcal{I'}|}.
    \]
    Now, from the inclusion-exclusion principle and the binomial theorem we get:
    \[
    \begin{split}
    \left|\bigcup_{(i,j) \in \mathcal{I}} B(r, j, i) \right| &= \sum_{\emptyset \ne \mathcal{I'} \subseteq \mathcal{I}} (-1)^{|\mathcal{I}'|}\left|\bigcap_{(i,j) \in \mathcal{I'}} B(r, j, i)\right| = \sum_{m=1}^{|\mathcal{I}|} {|\mathcal{I}| \choose m} (-1)^m k^{L-m} = \\
    &= \left(1-\left(1-\frac{1}{k}\right)^{|\mathcal{I}|}\right)k^L = \left(1-\left(1-\frac{1}{k}\right)^{|\mathcal{I}|}\right)|U_r|,
    \end{split}
    \]
    as needed. 
\end{proof}
In terms of our intuition for the meaning of $B(r, j, i)$, these two properties mean that if all provers agree on the assignment to the distinguished variables (i.e., the verifier strongly accepts) then the corresponding $k$ sets $B(r,j,i)$ cover $U_r$, and if the verifier does not weakly accept then they only cover a constant fraction of it, 
which is reminiscent of what we need.

We are now ready to define $h:A\rightarrow 2^U$, defined as follows:
\[
h((q, a, i)) = \bigcup_{r\in R: q_i(\phi, r) = q} B(r, \rho(\phi, r, q, a), i),
\]
where $\rho(\phi, r, q, a)$ denotes the assignment on the distinguished variables induced by answer $a$ (as the verifier extracts in step 5). This definition of $h((q, a, i))$  is the natural formalization of our intuition for $(q, a, i) \in A$  representing the case where prover $P_i$ answers question $q$ with answer $a$ and $B(r, j, i)$ representing the case where under random state $r$ of the verifier, the assignment to the distinguished variables induced by $P_i$'s answer is $j$. Per the definition of a normalized unweighted coverage function, $f:2^A\rightarrow [0,1]$ is now defined as follows: 
\[
f(S)
= \frac{1}{|U|}\left|\bigcup_{(q,a,i)\in S} h((q,a,i))\right|.
\]

\subsection{Proof of Lemma~\ref{lemma: coverage reduction}} \label{subsec: coverage reduction proof}
To prove Lemma~\ref{lemma: coverage reduction} we need to prove that $k'$ and $f$ satisfy conditions (1), (2) and (3) in the lemma.
    
    \textbf{Proof of condition (1):} Let $(q, a, i)\in A$, now:
    \[
    \begin{split}
    f((q,a,i)) &= \frac{|h((q,a,i))|}{|U|} = \frac{1}{|U|}\left|\biguplus_{r\in R: q_i(\phi, r) = q} B(r, \rho(\phi, r, q, a), i)\right| \\
    &=  \frac{1}{|U|} \left|\biguplus_{r\in R: q_i(\phi, r) = q} \{u\in U \mid u_{\rho(\phi, r, q, a)} = i \land u_{L+1} = r \}\right| \\
    &= \frac{1}{|U|} \sum_{r\in R: q_i(\phi, r) = q} |\{u\in U \mid u_{\rho(\phi, r, q, a)} = i \land u_{L+1} = r \}| = \frac{1}{|U|} \sum_{r\in R: q_i(\phi, r) = q} k^{L-1} \\
    &= \frac{1}{|U|} \cdot \frac{|R|}{|Q|} k^{L-1} = \frac{1}{|U|}\cdot \frac{|U|}{|Q|k} =\frac{1}{k'},
    \end{split}
    \]
    where for the second equality we used that by definition, the sets $B(r,j,i)$ are disjoint for different values of $r$.

    \textbf{Proof of condition (2):} Assume that $\phi$ is satisfiable, and let $\rho$ be a satisfying assignment (to all variables). Denote by $\rho_r$ the assignment induced by $\rho$ to the distinguished variables chosen by randomness $r$. For any $q\in Q$, denote by $a(q)$ the answer to question $q$ according to assignment $\rho$. We claim that 
    \[
    S = \{ (q, a(q), i) \mid q\in Q, i\in [k]\}
    \]
    satisfies condition (2) of Lemma~\ref{lemma: coverage reduction}. It is clear that $|S| = k|Q| = k'$ as needed. To show that $f(S) = 1$ we show that $\bigcup_{(q, a, i) \in S} h((q,a,i))=U$. Indeed:
    \[
    \begin{split}
    \bigcup_{(q, a, i) \in S} h((q,a,i)) &= 
    \bigcup_{(q,a, i)\in S} \bigcup_{r\in R: q_i(\phi, r)=q} B(r, \rho(\phi, r, q, a), i) = \bigcup_{(q,a, i)\in S} \bigcup_{r\in R: q_i(\phi, r)=q} B(r, \rho_r, i) \\
    &= \bigcup_{q\in Q, i \in [k]} \bigcup_{r\in R: q_i(\phi, r)=q} B(r, \rho_r, i) =  
    \bigcup_{r\in R, i\in [k]} \bigcup_{q\in Q: q_i(\phi, r)=q} B(r, \rho_r, i)
    \\
    &= \bigcup_{r\in R, i\in [k]} B(r, \rho_r, i) = \bigcup_{r\in R} U_r = U ,
    \end{split}
    \]
    where the second to last equality  follows  from Claim~\ref{claim:covering U_r}. Thus,
    \[
    f(S) = \frac{1}{|U|}\left|\bigcup_{(q,a,i)\in S} h((q,a,i))\right| = \frac{1}{|U|}|U|= 1.
    \]
    This concludes the proof of condition (2). We now turn to the proof of condition (3) of Lemma~\ref{lemma: coverage reduction}.
    
    \textbf{Proof of condition (3):} Assume that at most a $(1-\varepsilon')$-fraction of the clauses of $\phi$ are satisfiable, and assume towards contradiction the existence of a set $S\subseteq A$ of size $\beta k'$ where $\beta < M$ s.t $f(S) > 1-e^{-\beta} + \varepsilon$. We will show a contradiction for large enough $k,\ell$ (that depend only on $\varepsilon, M$), thus proving condition~(3).

    The intuition for the contradiction is that the set $S$ defines a strategy for the provers; namely, prover $P_i$ answers question $q\in Q$ by picking an answer $a$ uniformly out of those answers that satisfy the condition $(q, a, i) \in S$ (If no answer satisfies this condition, $P_i$ answers arbitrarily). The construction is such that, when fixing randomness $r\in R$, if the provers don't agree on the assignment to the distinguished variables, the union of $B(r, \rho_i, i)$ can't be too big, which means that if $f(S)$ is large and $|S|$ is small, this strategy of the provers has a relatively high probability of causing the verifier to weakly accept.
    
    More formally, each prover $P_i$ answers question $q$ by uniformly picking an answer from the set 
    \[
        \{a\in \{0,1\}^{2l} \mid (q, a, i)\in S\}.
    \]
    With each pair $(q,i) \in Q\times [k]$ we associate a weight $\alpha_{q,i}$, which is defined as the size of the above set. If $\alpha_{q,i} = 0$ prover $P_i$ answers question $q$ arbitrarily.
    It is trivial that
    \begin{equation}\label{eq:|S|}        
    \sum_{q,i} \alpha_{q,i} = |S|.
    \end{equation}
    With each randomness $r\in R$ we associate a weight $\alpha_r = \sum_{i=1}^k \alpha_{q_i(\phi, r),i}$.

    We call $r$ good if the two following conditions hold:
    \begin{enumerate}
        \item $\alpha_r \le \frac{3\beta k}{\varepsilon}$.
        \item There exist $(q_1, a_1, i_1), (q_2, a_2, i_2) \in S$ s.t $i_1 \ne i_2$, $q_1 = q_{i_1}(\phi, r), q_2 = q_{i_2}(\phi, r)$ and $\rho(\phi, r, q_1, a_1) = \rho(\phi, r, q_2, a_2)$. 
    \end{enumerate} 
    Intuitively, if $r$ is good the provers have a ``good'' probability of getting the verifier to weakly accept. 
    More formally:
    \begin{lemma}\label{lem:good}
If the randomness $r\in R$ chosen by the verifier is good, then when the provers employ the above strategy, the verifier weakly accepts with probability at least $\left(\frac{\varepsilon}{3Mk}\right)^2$.
    \end{lemma}
    \begin{proof}
        Since $r$ is good, there exist $(q_1, a_1, i_1), (q_2, a_2, i_2) \in S$ s.t. $i_1 \ne i_2$, $q_1 = q_{i_1}(\phi, r), q_2 = q_{i_2}(\phi, r)$ and $\rho(\phi, r, q_1, a_1) = \rho(\phi, r, q_2, a_2)$. 
        The probability of prover $P_{i_1}$ answering $a_1$ is at least $\frac{1}{\alpha_r}$ (since it selects its answer uniformly out of at most $\alpha_{q_1, i_1} \le \alpha_r$ answers), the same is true for prover $P_{i_2}$ answering $a_2$. This means the probability of the verifier weakly accepting is at least 
        \[
        \left(\frac{1}{\alpha_r}\right)^2 \ge \left(\frac{\varepsilon}{3\beta k}\right)^2 \ge \left(\frac{\varepsilon}{3Mk}\right)^2,
        \]
        which concludes the proof.
    \end{proof}

Based on Lemma~\ref{lem:good}, to obtain a contradiction to Lemma~\ref{k-prover lemma}, it remains to show that the probability that $r$ is good is sufficiently large. This is cast in the following lemma.

    \begin{lemma}\label{lem:suf}
        For a sufficiently large $k$, the fraction of good $r$ is at least $\frac{\varepsilon}{3}$.
    \end{lemma}
    
    Before presenting the proof of Lemma~\ref{lem:suf}, we briefly provide some intuition. Our first observation is that the mean value of $\alpha_r$ for $r\in R$ is exactly $\beta k$. 
    This serves two purposes, the first is that it allows us to treat the condition $\alpha_r > \frac{3Mk}{\varepsilon}$ as negligible (as it is only met by at most a $\frac{\epsilon}{3}$-fraction of values of $r$). 
    We then assume by contradiction that the negation of condition (2) of $r$ being good is negligible (as in condition 2 is met by at most a $\frac{\epsilon}{3}$-fraction of values of $r$). 
    When neglecting values of $r$ that satisfy condition 2 of $r$ being good, we apply Claim~\ref{claim: bounding disjoint cover} to show that the set $S$ only covers $(1-(1-\frac{1}{k})^{\alpha_r}) |U_r|$ of $|U_r|$, at which point we use the concavity of the function $y(x) = (1-(1-\frac{1}{k})^x)$ to bound how much of $U$ is covered by $S$ with $y(E[\alpha_r]) \approx y(\beta k) \approx 1-e^{-\beta}$, producing our contradiction. 

    We now present the formal proof of Lemma~\ref{lem:suf}.
    
    \begin{proof}
        Denote by $R_{\text{Good}}$ the set of good values of $r$. Assume towards contradiction that $|R_{\text{Good}}| \le \frac{\varepsilon\cdot |R|}{3} $.
        Denote by $R_{\text{Large}}$ the set of values of $r$ with $\alpha_r > \frac{3\beta k}{\varepsilon}$, i.e.,
        \[
        R_{\text{Large}} = \left\{r\in R\mid \alpha_r > \frac{3\beta k}{\varepsilon} \right\}.
        \]
        Finally, denote by $g_S(r)$ the number of elements in $U_r$ covered by $S$, i.e.:
        \[
        g_S(r) = \left|\bigcup_{(q,a,i)\in S} h((q,a,i)) \cap U_r\right|.
        \]
        Note that for any $r\in R$, it holds that $g_S(r) \le  |U_r| = k^L$.
        
        Since $f(S) \ge 1-e^{-\beta} + \varepsilon$, we have 
        \[
        \sum_{r\in R} g_S(r) = f(S) \cdot |U| \ge (1-e^{-\beta} + \varepsilon) \cdot |R| \cdot k^L.
        \]
        The average value of $\alpha_r$ for $r\in R$ is exactly $\beta k$, since
        \[
        \frac{1}{|R|} \sum_{r\in R} \alpha_r = \frac{1}{|R|} \sum_{r\in R} \sum_{i=1}^k \alpha_{q_i(\phi, r),i} = \frac{1}{|R|} \sum_{q\in Q,i \in [k]} \frac{|R|}{|Q|} \alpha_{q,i} = \frac{|S|}{|Q|}  = \beta k,
        \]
        where second equality follows from the fact that there are exactly $|R|/|Q|$ random states that cause the verifier to send out question $q$ to prover $P_i$, the third equality is by Equation~\eqref{eq:|S|}, and the last equality is since $|S| = \beta k' = \beta k |Q|$.
        This gives us, by Markov's inequality, $\frac{|R_{\text{Large}}|}{|R|} \le \frac{\varepsilon}{3}$, since
        \[
        \frac{|R_{\text{Large}}|}{|R|} \le \frac{\beta k}{{3\beta k}/{\varepsilon}} = \frac{\varepsilon}{3}.
        \]
        We now have both $|R_{\text{Large}}|$ and $ |R_{\text{Good}}|$ are at most $  \frac{\varepsilon\cdot |R|}{3} $. This, together with the fact that $R_{\text{Large}}$ and $R_{\text{Good}}$ are disjoint, imply that
        \begin{equation} \label{ineq: lower bound coverage by rs}
        \begin{split}
        \sum_{r\in R \setminus (R_{\text{Large}} \cup R_{\text{Good}})} g_S(r) &= \sum_{r\in R} g_S(r) - \sum_{r\in R_{\text{Large}}} g_S(r) - \sum_{r\in R_{\text{Good}}} g_S(r) \\ 
        &\ge (1-e^{-\beta} + \varepsilon) \cdot |R| \cdot k^L - |R_{\text{Large}}|\cdot k^L - |R_{\text{Good}}|\cdot k^L \\
        &\ge \left(1-e^{-\beta} + \frac{\varepsilon}{3}\right) \cdot |R| \cdot k^L.
        \end{split}
        \end{equation}
        Additionally, the average value of $\alpha_r$ for $r \in R \setminus (R_{\text{Large}} \cup R_{\text{Good}})$ is bounded from above by $(1+\frac{2\varepsilon}{3})\cdot \beta k$ since
        \[
        \begin{split}
        \frac{1}{|R| - |R_{\text{Large}}| - |R_{\text{Good}}|}\sum_{r\in R \setminus (R_{\text{Large}} \cup R_{\text{Good}})} \alpha_r &\le  \frac{1}{1-\frac{|R_{\text{Good}}|}{|R|-|R_{\text{Large}}|}} \cdot \frac{1}{|R| - |R_{\text{Large}}|} \sum_{r\in R \setminus R_{\text{Large}}} \alpha_r \\
        &\le \frac{1}{1-\frac{\frac{\varepsilon}{3}|R|}{|R|-\frac{\varepsilon}{3}|R|}} \cdot \frac{1}{|R|} \sum_{r\in R} \alpha_r  
        = \frac{1-\frac{\varepsilon}{3}}{1-\frac{2\varepsilon}{3}}\cdot \beta k \\ & \le \left(1+\frac{2\varepsilon}{3}\right)\cdot \beta k,
        \end{split}
        \]
        where the first inequality is because we're only adding non-negative values to the sum. The second inequality is because the average value of $\alpha_r$ in $R \setminus R_{\text{Large}}$ is at most the average value of $\alpha_r$ in $R$, as we only removed elements with a large $\alpha_r$. The last inequality is because $\frac{1-x}{1-2x} \le (1+2x)$ in the range $x\in [0,\frac{1}{4}]$, and $\frac{\varepsilon}{3} \le \frac{1}{3e} < \frac{1}{4}$.

        Finally, we note that for any $r \notin R_{\text{Good}}$, from Claim~\ref{claim: bounding disjoint cover}, and since the number of sets of the form $B(r, j, i)$ that are included in the covering of $U_r$ by $S$ is $\alpha_r$, we have $g_S(r) = (1-(1-\frac{1}{k})^{\alpha_r})k^L$. Since the function $y(x) = (1-(1-\frac{1}{k})^x)$ is concave, the sum
        \[
        \sum_{r\in R \setminus (R_{\text{Large}} \cup R_{\text{Good}})} g_r(S) = \sum_{r\in R \setminus (R_{\text{Large}} \cup R_{\text{Good}})} \left(1-\left(1-\frac{1}{k}\right)^{\alpha_r}\right) k^L
        \]
        is maximized when all $\alpha_r$ are equal to the mean, and therefore is bounded from above by
        \begin{eqnarray*}\label{eq:lower_grs}
        \begin{split}
        \sum_{r\in R \setminus (R_{\text{Large}} \cup R_{\text{Good}})} g_r(S) &\le k^L \sum_{r\in R \setminus (R_{\text{Large}} \cup R_{\text{Good}})} \left(1-\left(1-\frac{1}{k}\right)^{(1+\frac{2\varepsilon}{3})\beta k }\right) \\
        &\le k^L \cdot |R| \cdot \left(1-\left(1-\frac{1}{k}\right)^{(1+\frac{2\varepsilon}{3})\beta k }\right).
        \end{split}
        \end{eqnarray*}

        By combining with  Inequality~\eqref{ineq: lower bound coverage by rs}, we get that   the difference function in the upper and lower bounds (divided by $k^L
       \cdot |R|$) of  $\sum_{r\in R \setminus (R_{\text{Large}} \cup R_{\text{Good}})} g_r(S)$ must satisfy $$\delta(\varepsilon,k) =  e^{-\beta} -\left(1-\frac{1}{k}\right)^{-(1+\frac{2\varepsilon}{3})\beta k} - \frac{\varepsilon}{3} \geq 0.$$
        However, for $\delta^\star(\varepsilon) = \lim_{k\rightarrow \infty }\delta(\varepsilon,k) = e^{-\beta} -e^{-(1+\frac{2\varepsilon}{3})\beta } - \frac{\varepsilon}{3} $ it holds that $\delta^\star(0)=0$, and the function $\delta^\star$ is strictly decreasing in $\varepsilon$ for $\varepsilon\in [0,1]$, thus for every $\epsilon>0$ there exists a large enough $k$ for which $\delta(\varepsilon,k) < 0$ which  is a contradiction.
    \end{proof}
To conclude the  proof of Lemma~\ref{lemma: coverage reduction},        we observe that by combining Lemma~\ref{lem:good} and Lemma~\ref{lem:suf} it is sufficient to choose a sufficiently large value of $\ell$  such that $\frac{\varepsilon}{3}\cdot \left(\frac{\varepsilon}{3Mk}\right)^2 > k^2 \cdot2^{-c\ell}$, which contradicts Lemma~\ref{k-prover lemma}. \qed
\section{A Discussion on Approaches for Hardness Results}\label{app:approach}
In this appendix, we discuss why one cannot extend the hardness results from  \citep{multi-agent-contracts} for the cases of XOS and subaditive success probability functions to show a constant hardness result for submodular functions. \citet{multi-agent-contracts} show that a better-than-constant approximation for XOS $f$ is not possible in polynomial time, even given access to both value and demand oracles (a demand oracle returns a set $S \subseteq A$ that maximizes $f(S) - \sum_{i\in S} p_i$ for every price vector $p=(p_1, \dots, p_n)$).

The idea behind their approach was to ``hide'' an especially good set $T\subseteq A$ within an otherwise symmetric success probability function. When incentivizing the agents $T$, the principal would have a significantly greater utility than any other set $S\subseteq A$, but due to the symmetry of $f$, one can't guess $T$ with a non-negligible advantage using a polynomial number of value and demand queries.

In this section we show that such an approach cannot be used in the case of submodular success probability functions. Generalizing their approach, we characterize it as defining a pseudo-symmetric success probability function $f$ (i.e., $f(S)$ is a function of $|S|$, except for a single set $T$ which may be better). We show a PTAS of the optimal contract for any psuedo-symmetric submodular success probability function, which shows this approach cannot produce inapproximability results for submodular $f$. Our claim is formalized in the following theorem:
\begin{theorem} \label{thm: psuedo-symmetric PTAS}
    Let $A=[n]$ be the set of agents, $f:2^A \rightarrow [0,1]$ be a psuedo-symmetric submodular success probability function, and let $c_1 \le \dots \le c_n$ be associated agent costs.
    The subset of agents $S \in \arg \max_{S\in C}{g(S)}$ is a $(1-\varepsilon)$-approximation of the optimal solution, where $C= \{ \emptyset\} \cup \{[i] \mid i \in [n]\} \cup \{S \subseteq A \mid |S| \le \frac{2}{\varepsilon}\}$.
\end{theorem}
\begin{proof}
     Let $A=[n]$ be the set of agents, $f:2^A \rightarrow [0,1]$ be a psuedo-symmetric submodular success probability function, and let $c_1 \le \dots \le c_n$ be associated agent costs.
     Let $T\subseteq A$ be a subset of agents, $h:\mathbb{N} \rightarrow [0,1]$, and $v\in [0,1]$ such that $f(S) = h(|S|)+v \cdot 1[S=T]$.
     
    Denote an optimal contract for the problem instance by $S^\star$ (i.e., $S^\star \in \arg \max_{S\subseteq A} g(S)$), and let $S \in \arg \max_{S\in C}{g(S)}$, where $C$ is the same as in the theorem statement, we need to show $g(S) \ge (1-\varepsilon) \cdot g(S^\star)$. 

    We first note that for any set $S'\subseteq A$ with size $|S'|\notin \{|T|, |T|+1\}$ it holds that 
    \[
    g(S') = \left(1-\sum_{i\in S'} \frac{c_i}{f(i \mid S' \setminus \{i\})}\right) f(S')= \left(1-\frac{1}{h(|S'|) - h(|S'|-1)}\sum_{i\in S'} c_i\right) h(|S'|).
    \]
    This implies (since $c_1 \le \dots \le c_n$) that for
    \[
    g([|S'|]) \ge g(S').
    \]
    
    Therefore, if $|S^\star| < |T|$ or $|S^\star| > |T|+1$, it holds that $g(S) \ge g([|S^\star|]) \ge g(S^\star)$, as needed. 
    
    Also, if $|S^\star| \le \frac{2}{\varepsilon}$, it holds that $S^\star \in C$, implying, from our definition of $S$, that $g(S) \ge g(S^\star)$, as needed.
    Finally, if $\frac{2}{\varepsilon} < |S^\star|$ and $ |S^\star| \in \{|T|, |T|+1\}$, from submodularity there exists a set $S'\subseteq S^\star$ of size $|S^\star|-2$ with $f(S') \ge \frac{|S'|}{|S^\star|} f(S^\star) \ge (1-\varepsilon) f(S^\star)$. Since $S' \subseteq S^\star$ we get
    \[
    \begin{split}
    g(S') &= \left(1-\sum_{i\in S'} \frac{c_i}{f(i \mid S' \setminus \{i\})}\right) f(S') 
    \ge \left(1-\sum_{i\in S^\star} \frac{c_i}{f(i \mid S^\star \setminus \{i\})}\right) f(S') \\ 
    &\ge  \left(1-\sum_{i\in S^\star} \frac{c_i}{f(i \mid S^\star \setminus \{i\})}\right) \cdot (1-\varepsilon) \cdot f(S^\star) = (1-\varepsilon)\cdot g(S^\star),
    \end{split}
    \]
    where the first inequality is from submodularity of $f$, 
    and since $S' \subseteq S^\star$. Thus, we have 
    \[
    g(S) \ge g([|S'|]) \ge g(S') \ge (1-\varepsilon)g(S^\star),
    \]
    as needed.
\end{proof}

Theorem~\ref{thm: psuedo-symmetric PTAS} clearly gives us a PTAS for the optimal contract when $f$ is a pseudo-symmetric submodular function, as an algorithm that simply computes $\arg \max_{S\in C} g(S)$ runs in polytime, as $|C| = O(n^{2/\varepsilon})$.

This proves that the inapproximability approach from \citep{multi-agent-contracts} doesn't work for submodular functions. It is also worth noting that this approach produces (when it applies) different hardness results than our own, since our results apply even when the algorithm is given a succinct representation of the success probability function, and has complete information about it. In other words - our approach shows computational hardness results, whereas the approach in \citep{multi-agent-contracts} shows information theoretical hardness.

\end{document}